\documentclass[12pt]{amsart}
\usepackage[margin=1in]{geometry}
\usepackage{amsfonts}
\usepackage{amssymb}
\usepackage[dvips]{graphics}
\usepackage{epsfig}
\usepackage{euscript}
\usepackage{color}
\usepackage{bbm}
\usepackage[colorlinks]{hyperref}
\hypersetup{
    linkcolor=red,          
    citecolor=blue,        
}

\newtheorem{thm}{Theorem}[section]

\newtheorem{lem}[thm]{Lemma}
\newtheorem{proposition}[thm]{Proposition}
\theoremstyle{definition}

\theoremstyle{remark}
\newtheorem{remark}[thm]{Remark}
\numberwithin{thm}{section}

\newcommand{\R}{{\mathord{\mathbb R}}}
\newcommand{\N}{{\mathord{\mathbb N}}}

\newcommand{\C}{{\mathord{\mathbb C}}}
\newcommand{\Z}{{\mathord{\mathbb Z}}}

\DeclareMathOperator{\I}{\mathbbm{1}}

\DeclareMathOperator{\pf}{pf}
\DeclareMathOperator{\tr}{tr}
\DeclareMathOperator{\Tr}{Tr}

\DeclareMathOperator{\diag}{diag}
\begin{document}

\title[Area law for the disordered XY chain]{A uniform area law for the entanglement of eigenstates in the disordered XY chain}

\author[H. Abdul-Rahman]{Houssam Abdul-Rahman$^1$}
\address{$^1$ Department of Mathematics\\
University of Alabama at Birmingham\\
Birmingham, AL 35294 USA}
\email{houssam@uab.edu}

\author[G. Stolz]{G\"unter Stolz$^1$}
\thanks{G.\ S.\ was supported in part by NSF grant DMS-1069320.}
\email{stolz@uab.edu}

\date{}

\maketitle




\begin{abstract}
We consider the isotropic or anisotropic XY spin chain in the presence of a transversal random magnetic field, with parameters given by random variables. It is shown that eigenfunction correlator localization of the corresponding effective one-particle Hamiltonian implies a uniform area law bound in expectation for the bipartite entanglement entropy of all eigenstates of the XY chain, i.e.\ a form of many-body localization at all energies. Here entanglement with respect to arbitrary connected subchains of the chain can be considered. Applications where the required eigenfunction correlator bounds are known include the isotropic XY chain in random field as well as the anisotropic chain in strong or strongly disordered random field.
\end{abstract}

\vspace{.5cm}


\vspace{.3cm}




\allowdisplaybreaks
\section{Introduction}

\subsection{Motivation: Towards Many-Body Localization}

We are interested here in a better understanding of mathematical characterizations of many-body localization (MBL) in interacting quantum systems. This phenomenon has recently received strong attention in theoretical physics and quantum information theory, see e.g. \cite{Bardarsonetal, Baskoetal, BauerNayak, BurrellOsborne,Canovietal, Friesdorfetal, Gornyietal, HuseOganesyan, OganesyanHuse, PalHuse, Serbynetal2014, Serbynetal2013a, Serbynetal2013b, VoskAltman, Znidaricetal} and references therein. It is generally described as the absence of {\it thermalization} or {\it self-equilibration} in a quantum many-body system, often due to the presence of disorder. For a detailed discussion of the current understanding of thermalization and many-body localization in the physics literature see also the recent survey \cite{GogolinEisert}, including the extensive list of references provided there. It is made clear in these works that there is no complete consensus yet on what physically constitutes thermalization and MBL, thus leaving a multitude of questions for further investigation and clarification.

We will focus here on some of the criteria which by now are well accepted to be necessary characteristics of MBL and on studying these criteria for relatively simple models where they can be proven rigorously. Generally, and quite roughly, these criteria say that an interacting system, in suitable regimes such as weak interaction or large disorder, has properties similar to those of a non-interacting many-body system. The eigenstates of the latter are product states and thus have vanishing entanglement and spatial correlations. Also, its dynamics are trivial with no information propagating between particles. For an interacting system in the MBL phase one thus expects rapid decay of correlations and small entanglement of eigenstates, as well as absence of information transport (i.e.\ no or slow propagation of particle group waves). Note that this is a true many-body concept and to be distinguished from classical Anderson localization for a single particle, say, in a disordered environment, where `localization' refers to the single particle configuration space. Significant differences between the many-body localized phase and single particle localization have been pointed out, e.g., in \cite{Bardarsonetal, Serbynetal2014}.

Let us stress here that the term MBL should generally be reserved for properties which hold uniformly in the number of particles in the system (e.g.\ in the sense that relevant constants are bounded in the particle number). In this sense, many-body localization is to be distinguished from few-body localization, such as the known rigorous results for the $N$-particle Anderson model \cite{ChulaevskySuhov1, ChulaevskySuhov2, AizenmanWarzel, KleinNguyen}, which do not yet allow uniform control in the number of electrons. In particular, results expected by physicists for the many-body Anderson model, such as MBL at low electron density or weak interaction strength, e.g.\ \cite{Baskoetal,Gornyietal}, can not yet be shown rigorously.

Rigorous mathematical results on localization properties of disordered many-body systems are so far essentially restricted to two models: Disordered harmonic oscillator systems \cite{NSS1, NSS2} and the XY spin chain in random field \cite{KleinPerez, HamzaSimsStolz, SimsWarzel, PasturSlavin}. These models are equivalent to free Boson systems and free Fermion systems, respectively, and thus can be studied in terms of an effective one-particle Hamiltonian. As a consequence, it is possible to deduce results on MBL from known localization properties of one-particle Hamiltonians such as the Anderson model. Of course, a long term goal must be to develop methods which allow to go beyond these simple models. In particular, an important challenge is to develop mathematical methods to study the disordered XXZ or Heisenberg chain, which reduce to the physically more interesting situation of {\it interacting} Fermion systems, see e.g.\ \cite{Bardarsonetal, BauerNayak, OganesyanHuse, PalHuse, Serbynetal2013a} for numerical results. Some progress towards understanding MBL for larger classes of spin systems in the presence or strong disorder has been made in \cite{Imbrie}, which, however, still needs to make an unproven assumption to limit the amount of level attraction in the system. Mathematically, describing the phenomenon of many-body localization and fully proving it for important classes of physical examples is a wide open field (even wider than for physics).

Our more modest goal is here is to contribute to a more complete understanding of MBL properties of the two models indicated above, harmonic oscillators and the XY chain.

For oscillator systems the reduction to an effective one-particle Hamiltonian is rather straightforward. In particular, the reduction works in any dimension and does not affect locality properties of the system. This has allowed to verify a quite complete list of MBL properties for disordered oscillator systems, including a zero-velocity Lieb-Robinson bound on information transport as well as exponential decay of correlations and area-law-type entanglement bounds for ground and thermal states \cite{NSS1, NSS2} (in each case requiring disorder averaging).

The reduction of the XY model via the Jordan-Wigner transform to a free Fermion system is only possible for a one-dimensional chain of spins and, as an additional difficulty, introduces non-locality. Therefore rigorously known MBL properties for the XY model are more limited. A first contribution was made in \cite{KleinPerez} which established exponential decay of certain ground state correlations for the isotropic XY chain in random exterior field, using localization properties of the Green function of the effective Hamiltonian, in this case the Anderson model, proven via multiscale analysis. More general results have recently been obtained by Sims and Warzel \cite{SimsWarzel}. Using localization of eigenfunction correlators of the underlying one-particle Hamiltonian, these authors find exponential decay of stationary as well as time-dependent correlations for larger classes of states, including general eigenstates as well as thermal states, for systems such as the XY chain which can be mapped to free Fermions.

Localization of eigenfunction correlators is also the key property of the effective Hamiltonian used in \cite{HamzaSimsStolz} to show a zero-velocity Lieb-Robinson bound for the XY chain in random field, again after disorder averaging. An argument in \cite{HamzaSimsStolz}, valid for a very general class of quantum spin systems shows that this implies exponential decay of ground state correlations up to a logarithmic correction in the size of the ground state gap (the more specific arguments used in \cite{SimsWarzel} show that no such correction is required for the random XY chain). This in turn, by another general observation, leads to an area law for the ground state entanglement \cite{BrandHor13, BrandHor15}. Note, however, that \cite{BrandHor13, BrandHor15} consider deterministic systems and that some further analysis would be required to show that the relation ``exponential decay of correlations implies area law'' carries over to the disorder averaged quantities.

Here we prove an area law for the disorder averaged entanglement of {\it all} eigenstates of the XY chain in random field. In fact, we state this as a conditional result: If the corresponding effective one-particle Hamiltonian has localized eigenfunction correlators, then the many-body eigenstates satisfy an area law, uniform in energy (i.e.\ in the labeling of eigenstates), see Section~\ref{sec:mainresult} for exact statements. Applications of this result, see Section~\ref{sec:Applications}, include the isotropic XY chain but also some anisotropic cases.

To our knowledge this is the first result on many-body localization for eigenstates of a disordered spin chain which goes beyond the ground state, the latter physically corresponding to the case of zero temperature. In fact, the uniform area law for all eigenstates found here corresponds to the infinite temperature limit where all eigenstates are equally probable. That the validity of an area law for all (or at least a large fraction of all) eigenstates is a necessary characteristic of the MBL phase was recently stressed by Bauer and Nayak in \cite{BauerNayak}. These authors provide a formal definition of the MBL phase, reflecting the idea of Fock space localization proposed in \cite{Baskoetal}, and argue that an area law for the entanglement of eigenstates should generally follow as a consequence. It is natural to expect that localization at infinite temperature implies localization at finite temperature. However, our methods (in particular the use of the von Neumann entropy of reduced pure states to measure entanglement) do not directly extend to give useful entanglement bounds for mixed Gibbs states.

Our result is closely related to recent work by Pastur and Slavin \cite{PasturSlavin} who establish an area law for the ground state of disordered free Fermion systems. In fact, by mapping the XY chain to a free Fermion system via the Jordan-Wigner transform, we use the argument from \cite{PasturSlavin} as one of the key ingredients to our proof. We extend this argument to general eigenstates and, to include the anisotropic case, also consider more general Fermion systems which do not conserve the particle number, leading to random block operators as effective Hamiltonians. It may also be worth noting that we consider entanglement with respect to arbitrary subchains $\Lambda_1$ of the chain $\Lambda = \{1,\ldots,n\}$ in Theorem~\ref{mainthm} below (and not just left or right ends of $\Lambda$). This requires to work with a suitably chosen local version of the Jordan-Wigner transform, introduced in Section~\ref{corredstates} below.

\vspace{.3cm}

\noindent {\bf Acknowledgement:} It is our great pleasure to acknowledge many helpful discussions with Bruno Nachtergaele and Robert Sims on the contents of this work. G.\ S.\ would also like to thank the Isaac Newton Institute for Mathematical Sciences, Cambridge, for support and hospitality during the program {\it Periodic and Ergodic Spectral Problems} in Spring 2015 where part of the work on this paper was done.



\subsection{The Main Result} \label{sec:mainresult}

We consider the anisotropic XY spin chain in transversal magnetic field, given by the self-adjoint Hamiltonian
\begin{equation} \label{eq:anisoxychain}
H := - \sum_{j=1}^{n-1} \mu_j [ (1+\gamma_j) \sigma_j^X \sigma_{j+1}^X + (1-\gamma_j) \sigma_j^Y \sigma_{j+1}^Y] - \sum_{j=1}^n \nu_j \sigma_j^Z
\end{equation}
in ${\mathcal H} = \bigotimes_{j\in \Lambda} \C^2$, where $\Lambda=\{1,2,\ldots, n\}$ for an arbitrary positive integer $n$. By $\sigma_j^X$, $\sigma_j^Y$ and $\sigma_j^Z$ we denote the standard Pauli matrices acting on the $j$-th component of the tensor product.

The parameters $\mu_j$, $\gamma_j$ and $\nu_j$ describe the interaction strength, anisotropy and field strength, respectively, and we will think of them as the first $n$ components of sequences of real-valued random variables indexed by $j\in \N$. While stronger assumptions will be needed in applications, for our main result we only will assume their uniform boundedness, i.e.\ that there exists $C<\infty$ such that
\begin{equation} \label{assump}
\sup_{j\in \N}\; \left( |\mu_j| + |\gamma_j|+ |\nu_j| \right) \le C \quad \mbox{almost surely}.
\end{equation}

We will also assume that the disorder guarantees that
\begin{equation} \label{simple}
H \;\;\mbox{almost surely has simple spectrum},
\end{equation}
in the sense that all its eigenvalues are non-degenerate. In the applications discussed in Section~\ref{sec:Applications} below this will follow from an argument provided in Appendix~\ref{App:Non-degeneracy}.

As known since the work of Lieb, Schultz and Mattis \cite{LSM}, the XY chain Hamiltonian $H$ can be mapped to a free Fermion system on ${\mathcal H}$ by the Jordan-Wigner transform. As a result, the explicit diagonalization of $H$  reduces to the diagonalization of an effective one-particle Hamiltonian which can be written in form of the symmetric $2\times 2$-block Jacobi matrix
\begin{equation} \label{effHam}
M := \left( \begin{array}{cccc} -\nu_1 \sigma^Z & \mu_1 S(\gamma_1) & & \\ \mu_1 S(\gamma_1)^t & -\nu_2 \sigma^Z & \ddots & \\ & \ddots & \ddots & \mu_{n-1} S(\gamma_{n-1}) \\ & & \mu_{n-1} S(\gamma_{n-1})^t & -\nu_n \sigma^Z \end{array} \right),
\end{equation}
where $\sigma^Z = \begin{pmatrix} 1 & 0 \\ 0 & -1 \end{pmatrix}$  is the third Pauli matrix and
\begin{equation}
S(\gamma) = \begin{pmatrix} 1 & \gamma \\ -\gamma & -1 \end{pmatrix}.
\end{equation}
For some more details on the reduction of $H$ to $M$ see Section~\ref{sec:JordanWigner} below.

Our main result says that a suitable form of localization of the effective one-particle Hamiltonian $M$ implies a uniform area law for the bipartite entanglement entropy of all eigenstates of the $n$-spin Hamiltonian $H$. To introduce the latter, consider an arbitrary subinterval $\Lambda_1 = \{r, \ldots, r+\ell-1\}$ of $\Lambda$ and the bipartite decomposition ${\mathcal H} = {\mathcal H}_1 \otimes {\mathcal H}_2$, where
\begin{equation} \label{bipartite}
 {\mathcal H}_1 = \bigotimes_{j\in \Lambda_1} \C^2 \quad \mbox{and} \quad {\mathcal H}_2 = \bigotimes_{j\in \Lambda \setminus \Lambda_1} \C^2.
\end{equation}
For a pure state $\rho$ in ${\mathcal H}$ the entanglement entropy $\mathcal{E}(\rho)$ with respect to the subsystem $\Lambda_1$ is given by the von Neumann entropy of the reduced state $\rho_1 = \Tr_{{\mathcal H}_2} \rho$,
\begin{equation}
\mathcal{E}(\rho) = \mathcal{S}(\rho_1) := - \Tr g(\rho_1),
\end{equation}
where
\begin{equation}
g(x) = \left\{ \begin{array}{ll} x \log x, & 0<x \le 1, \\ 0, & x=0. \end{array} \right.
\end{equation}
We choose $\log$ to denote the natural logarithm, as opposed to the dyadic logarithm used in this context in the information theory literature. The distinction is irrelevant for our work, as we will not keep track of universal constants.

With $\mathbb{E}(\cdot)$ denoting the disorder average, we will prove

\begin{thm} \label{mainthm}
Assume (\ref{assump}) and (\ref{simple}), as well as the existence of $\beta >2$ and $C<\infty$ such that
\begin{equation} \label{ecloc}
\mathbb{E} \left( \sup_{|g|\le 1} \|g(M)_{jk} \| \right) \le \frac{C}{1+|j-k|^{\beta}}
\end{equation}
for all $n\in \N$ and $1\le j,k \le n$.

Then there exists $\tilde{C} <\infty$ such that
\begin{equation} \label{arealaw}
\mathbb{E} \left( \sup_{\psi} \mathcal{E}(\rho_{\psi}) \right) \le \tilde{C}
\end{equation}
for all $n$, $r$ and $\ell$ with $1\le r \le r+\ell-1 \le n$. In (\ref{arealaw})  the supremum is taken over all normalized eigenfunctions $\psi$ of $H$ and $\rho_{\psi} = |\psi \rangle \langle \psi|$ is the orthogonal rank-one projection onto $\C \psi$.
\end{thm}

In (\ref{ecloc}), $g: \R\to \C$ may be any function such that $|g(x)|\le 1$ uniformly, and $g(M)$ is defined via the functional calculus for symmetric matrices. Using the block matrix notation from (\ref{effHam}) above, the matrix elements $g(M)_{jk}$ are understood as $2\times 2$-matrices. While any norm $\|\cdot\|$ on the $2\times 2$-matrices could be used in (\ref{ecloc}), we will work with the Euclidean matrix norm for the sake of definiteness.

Bounds of the form (\ref{ecloc}) are generally referred to as {\it localization of eigenfunction correlators}. Their use in single particle localization theory originates from \cite{Aizenman94}, see also \cite{Stolz} for an introductory survey.

In particular, considering the functions $g_t(x) = e^{-itx}$, $t\in \R$, the bound (\ref{ecloc}) includes dynamical localization, uniform in time, for solutions of the Schr\"odinger equation $\psi'(t) = -iM \psi(t)$. We note that requiring only power law decay for some $\beta>2$ in (\ref{ecloc}) is a relatively weak form of eigencorrelator localization. In applications where (\ref{ecloc}) can be verified, one usually has much faster decay in $|j-k|$, see Section~\ref{sec:Applications} below.

The ``area law'' (\ref{arealaw}), giving an upper bound for the entanglement proportional to the surface area of the subsystem $\Lambda_1$ (as opposed to the elementary bound in terms of its volume), is uniform not only in the size of the system $\Lambda$ and subsystem $\Lambda_1$, but also applies uniformly to {\it all} eigenstates of $H$. It thus yields a form of many-body localization for $H$ at {\it arbitrary energy}.

Note that by a general result due to Hastings \cite{Hastings}, an area law holds for the ground state of one-dimensional spin systems with a uniform (in volume) ground state energy gap. Not only is our work not limited to the ground state, but we also do not need any explicit assumption of the size of spectral gaps. Instead, we exploit single particle localization of the effective Hamiltonian $M$ as the mechanism to prove an area law. Of course, bounds on level statistics can be considered as being an implicit part of the underlying single particle theory. A price we have to pay for not having deterministic gaps between energy levels is the need to include disorder averages in (\ref{ecloc}) and (\ref{arealaw}).

\subsection{Applications} \label{sec:Applications}

Here we discuss some concrete examples, i.e.\ assumptions on the random variables $\mu_j$, $\gamma_j$ and $\nu_j$ which guarantee (\ref{assump}), (\ref{simple}) and, most importantly, one-particle localization of $M$ in the form (\ref{ecloc}).

We will focus on the case of random magnetic field, where the corresponding one-particle localization properties are best understood, i.e.\ we choose $\mu_j = \mu \in \R \setminus \{0\}$ and $\gamma_j = \gamma \in \R$ to be constant, while the $\nu_j$ are i.i.d.\ random variables with compactly supported distribution $\rho$. Thus (\ref{assump}) holds. Assuming, moreover, that $\rho$ is absolutely continuous $d\rho(\nu) = h(\nu)d\nu$, it follows from Proposition~\ref{prop:nondeg} in Appendix~\ref{App:Non-degeneracy} that (\ref{simple}) is satisfied as well.

The localization property (\ref{ecloc}) is known for the following cases.

(i) Let $\gamma=0$, so that
\begin{equation}
H= -\mu \sum_j (\sigma_j^X \sigma_{j+1}^X + \sigma_j^Y \sigma_{j+1}^Y) - \sum_j \nu_j \sigma_j^Z
\end{equation}
is the isotropic XY chain in random field. Then all $2\times 2$-matrix-entries of $M$ are diagonal and thus $M$ decomposes into $A \oplus (-A)$, where
\begin{equation}
A = \left( \begin{array}{cccc} -\nu_1 & \mu & & \\ \mu & \ddots & \ddots & \\ & \ddots & \ddots & \mu \\ & & \mu & -\nu_n \end{array} \right)
\end{equation}
is the Anderson model on the finite interval $\Lambda$. If the density $h$ is bounded and compactly supported, then it is known that, for some $C<\infty$ and $\eta>0$,
\begin{equation}
\mathbb{E} \left(\sup_{|g|\le 1} |g(A)_{jk}| \right) \le C e^{-\eta |j-k|},
\end{equation}
see e.g.\ \cite{Stolz}. This readily implies (\ref{ecloc}) with exponential decay in $|j-k|$.

(ii) For the anisotropic case $\gamma \not=0$, which does not reduce to the Anderson model, localization properties of one-particle Hamiltonians given by block Jacobi matrices of the form (\ref{effHam}) have been studied more recently.

A result by Elgart, Shamis and Sodin \cite{ESS} covers the large disorder case: If the i.i.d.\ random variables $\nu_j^{(0)}$, $j\in \N$, have bounded compactly supported density, and if $\nu_j = \lambda \nu_j^{(0)}$ for $\lambda>0$ sufficiently large, then (\ref{ecloc}) holds with exponential decay in $|j-k|$. This is done in \cite{ESS} through an adaptation of the fractional moments method to a class of random block operators which includes our model (\ref{effHam}).

If the magnetic field is strong enough to create a spectral gap for $M$ around energy $E=0$ (uniform in the volume $n$ and the disorder), then we can apply a result from \cite{ChapmanStolz}: Suppose that for some $C>|\mu|$ it holds that either $\nu_j \ge C$ for all $j\in \N$ or $\nu_j \le -C$ for all $j\in \N$, resulting in the spectral gap $(-(C-|\mu|), C-|\mu|)$ for $M$. Then for every $\xi<1$ there is $C=C(\xi)$ and $\eta = \eta(\xi)>0$ such that
\begin{equation} \label{subexp}
\mathbb{E} \left( \sup_{|g|\le 1} \|g(M)_{jk} \| \right) \le C e^{-\eta |j-k|^{\xi}}.
\end{equation}
This is essentially what is shown in the proof of Theorem~7.2 of \cite{ChapmanStolz}, while the results there are only stated in terms of dynamical localization for $M$, i.e.\ for the functions $g_t(x) = e^{-itx}$ (but the argument covers general $|g|\le 1$).

As discussed in \cite{ChapmanStolz}, it remains an open problem if a bound such as (\ref{subexp}) (or at least (\ref{ecloc})) can be proven for $M$ without the assumption of a zero-energy gap. A problem arises from the fact that the transfer matrix group (F\"urstenberg group) associated with $M$ at $E=0$ is not irreducible, so that many of the available dynamical systems tools for proving one-dimensional localization don't apply.

\subsection{Notes on Content and Related Work}

The rest of this paper is devoted to the proof of Theorem~\ref{mainthm}. As in most works on the XY chain, the starting point is the representation of the XY Hamiltonian $H$, via the Jordan-Wigner transform, as a quadratic form in a set of Fermionic operators, see Section~\ref{sec:JordanWigner}. A key tool  is that eigenstates of the latter satisfy Wick's rule (sometimes also referred to as being {\it quasi free} or {\it gaussian}) and thus are fully determined by their correlation matrices with respect to the Jordan-Wigner Fermionic operators.

More specifically, the crucial formula for determining the bipartite entanglement entropy of eigenstates is the trace identity (\ref{redtocormat}) in Lemma~\ref{lem:quasifree} below.  We will use it for the case where the state $\rho$ in (\ref{redtocormat}) is the reduction of an eigenstate of $H$ to the subsystem $\mathcal{H}_1$ (which will ultimately be justified in Lemma~\ref{thm:Wicks:rho1}). Thus the left hand side of (\ref{redtocormat}) becomes the entanglement entropy of an eigenstate. We thus need to calculate the correlation matrix on the right hand side of (\ref{redtocormat}), which will be accomplished in Lemma~\ref{rhoCcorrelation} and Theorem~\ref{thm:EErule}: Lemma~\ref{rhoCcorrelation} provides the central connection between the many-body Hamiltonian $H$ and the effective one-particle Hamiltonian $M$, saying that the correlation matrices of eigenstates of $H$ are given by spectral projections of $M$. Moreover, the proof of Theorem~\ref{thm:EErule} shows that correlation matrices of reduced states are given by $2\ell \times 2\ell$-restrictions of these spectral projections.

Equipped with these tools we can then complete the proof of the area law (\ref{arealaw}) in Section~\ref{sec:arealaw}. This is accomplished by using an extension of an argument from \cite{PasturSlavin} to reduce the claim to the eigencorrelator localization bound (\ref{ecloc}).

The trace identity (\ref{redtocormat}) for Fermion systems seems to originate from \cite{VLRK}, see also \cite{LRV} for a more detailed presentation and \cite{EisertCramerPlenio} for a survey of subsequent results. In these works (\ref{redtocormat}) was used as a tool in the study of the ground state entanglement for the XY chain in {\it constant} magnetic field,
\begin{equation}
- \sum_{j=1}^{n-1} [ (1+\gamma) \sigma_j^X \sigma_{j+1}^X + (1-\gamma) \sigma_j^Y \sigma_{j+1}^Y] - \nu \sum_{j=1}^n  \sigma_j^Z.
\end{equation}
In particular, the dimension reduction from $2^n$ to $2n$ accomplished by (\ref{redtocormat}) (or $2^{\ell}$ to $2\ell$ for the subsystem) allowed numerical predictions for the limit of large subsystems $\Lambda_1$ (and in the thermodynamic limit of an infinite chain $\Lambda$). These were later proven rigorously, see \cite{Itsetal} for most complete results as well as references therein. Generally, it was found that the ground state entanglement for the anisotropic chain remains bounded in the size $\ell$ of the subsystem, i.e. satisfies an area law. Phase transitions with divergent entanglement appear for $\nu=2$ and for the isotropic chain $\gamma=0$. In particular, \cite{JinKorepin} shows that in the latter case the ground state entanglement grows as $\log \ell$.

Our work here differs in several significant respects from all previous works on entanglement in the XY chain and the related Fermion systems: We consider general eigenstates of the XY chain and not only ground states here, we work in finite volume (proving bounds which hold uniformly in the volume) rather than in the thermodynamic limit of an infinite chain, and, finally, use a modification of the Jordan-Wigner transform to be able to consider entanglement with respect to more general subsystems. Partly for these reasons, but also to make a number of tools from the physics literature accessible to a broader audience, we include a self-contained presentation of relevant parts of the theory of finite Fermion systems here. This includes thorough discussions of correlation matrices, Bogoliubov transforms and quasi free states. Some of this is contained in Section~\ref{sec:EEE}, with more background collected in Appendix~\ref{App:Bogo}.

\section{Reduction of $H$ to $M$ via Jordan-Wigner} \label{sec:JordanWigner}

The importance of the XY chain as a model in the theory of quantum spin systems goes back to the work of Lieb, Schultz and Mattis \cite{LSM}, where it was shown that the XY chain with constant coefficients (and initially without magnetic field) is an exactly solvable model. Their argument proceeds by using the Jordan-Wigner transform to reduce the XY chain to a free Fermion system. For the last half century this has turned the XY chain into a canonical toy model for quantum spin systems, frequently used as a first example to illustrate new concepts.

It was understood that the methods of Lieb, Schultz and Mattis can be extended to include magnetic fields and to allow for variable coefficients $\mu_j$, $\gamma_j$ and $\eta_j$. In the latter case the model is not exactly solvable, rather the diagonalization of $H$ can be reduced to the diagonalization of the effective Hamiltonian $M$.

Below we briefly summarize this reduction. A similar account with somewhat more detail (but different sign conventions) is provided in Section~3.1 of \cite{HamzaSimsStolz}.

The local lowering and raising operators are
\begin{equation}\label{eq:as}
a := \frac{1}{2}(\sigma^X -i \sigma^Y) = \left( \begin{array}{cc} 0
& 0 \\ 1 & 0 \end{array} \right),\ \text{and}\ a^* = \frac{1}{2}(\sigma^X +i \sigma^Y) = \left( \begin{array}{cc} 0
& 1 \\ 0 & 0 \end{array} \right),
\end{equation}
whose actions on the $j$-th spin in $\mathcal H$ will be denoted by $a_j$ and $a_j^*$.

The {\it Jordan-Wigner transform} refers to the operators
\begin{equation} \label{eq:cjs}
c_1 := a_1, \quad c_j := \sigma_1^Z \ldots \sigma_{j-1}^Z a_j, \quad j= 2,\ldots,n,
\end{equation}
which satisfy the {\it canonical anti-commutation relations} (CAR)
\begin{equation} \label{CARproperties}
\{c_j, c_k^*\} = \delta_{jk}\I, \quad \{c_j, c_k\} = \{c_j^*, c_k^*\} = 0 \quad \mbox{for all $j, k =
1,\ldots,n$}.
\end{equation}
Following the more general convention used in Appendix~\ref{App:Bogo} below, we will refer to
\begin{equation} \label{JWsystem}
\mathcal{C}= (c_1, c_1^*, \ldots,\ldots, c_n, c_n^*)^t
\end{equation}
as the {\it Jordan-Wigner Fermionic system}.

The XY chain $H$ can be expressed in terms of the Jordan-Wigner operators as
\begin{equation} \label{eq:Crep}
H =   \mathcal{C}^* M \mathcal{C}.
\end{equation}
Here the right hand side is interpreted in the sense of matrix multiplication of the column $\mathcal{C}$, the row $\mathcal{C}^* = (c_1^*, c_1 \ldots, c_n^*, c_n)$, and the scalar $2n\times 2n$-matrix $M$ given in (\ref{effHam}). If $P$ is a permutation matrix which maps the canonical basis vectors $e_1,\ldots,e_{2n}$ of $\mathbb{C}^{2n}$ to $e_1,e_{n+1},e_2,e_{n+2},\ldots,e_n,e_{2n}$, then
\begin{equation}\label{blockmatrix}
P M P^t=: \tilde{M} =\begin{pmatrix}
            A & B \\
            -B & -A \\
          \end{pmatrix},
\end{equation}
where
\begin{equation}
A=\left( \begin{matrix} -\nu_1 & \mu_1 & & \\ \mu_1 & \ddots & \ddots & \\ & \ddots & \ddots & \mu_{n-1} \\ & & \mu_{n-1} & -\nu_n \end{matrix} \right),
\;\;B= \left( \begin{array}{cccc} 0 &\gamma_1 \mu_1 &  & \\ -\gamma_1 \mu_1 & \ddots & \ddots&   \\ & \ddots & \ddots & \gamma_{n-1} \mu_{n-1} \\ & & -\gamma_{n-1} \mu_{n-1} & 0 \end{array} \right).
\end{equation}

We see $A^* = A^t = A$ and $B^* = B^t = -B$, and thus $\tilde{M}^* = \tilde{M}^t = \tilde{M}$.

Transforming with the unitary $\begin{pmatrix} 0 & \I \\ \I & 0 \end{pmatrix}$, we see that $\tilde{M}$ is unitarily equivalent to $-\tilde{M}$, and thus, in particular, has spectrum symmetric to zero. It can be diagonalized by an orthogonal matrix $\hat{W}$ Let
\begin{equation} \label{Wmatrix}
\hat{W} := \frac{1}{2} \left( \begin{array}{cc} V+U & V-U \\ V-U & V+U \end{array} \right),
\end{equation}
where $U$ and $V$ are real orthogonal matrices associated with the singular value decomposition of $S:=A+B$ via
\begin{equation}\label{eq:UV}
USV^t=\Lambda:=\diag\{\lambda_1,\lambda_2,\ldots,\lambda_n\}.
\end{equation}
Here $0\leq \lambda_1\leq \lambda_2\leq \ldots\leq \lambda_n$ are the singular values of $S$, i.e.\ the eigenvalues of $(S^*S)^{\frac{1}{2}}$, counted with multiplicity. Calculations show that
\begin{equation}\label{diagM}
\widetilde{M}=\hat{W}^t\begin{pmatrix}
                         \Lambda & 0 \\
                         0 & -\Lambda \\
                       \end{pmatrix}\hat{W}, \quad W M W^t= \bigoplus_{j=1}^n\begin{pmatrix}
          \lambda_j & 0 \\
          0 & -\lambda_j \\
        \end{pmatrix},
\end{equation}
where $W=P^t \hat{W} P$ is Bogoliubov, i.e.\ orthogonal and satisfies (\ref{eq:Bogo:WJ}).
Let
\begin{equation} \label{Bsystem}
\mathcal{B}:=W\mathcal{C}.
\end{equation}
By Lemma~\ref{thm:Bogo:transf} this is a Fermionic system and
\begin{eqnarray} \label{eq:anisotropicFermi}
H & = & \mathcal{C}^* M \mathcal{C} = \mathcal{B}^* W M W^t \mathcal{B} = \mathcal{B}^* \bigoplus_{j=1}^n\left( \begin{array}{cc} \lambda_j & 0 \\ 0 & -\lambda_j \end{array} \right) \mathcal{B}  \\
& = & \sum_{j=1}^n \lambda_j (b_j^* b_j - b_j b_j^*)
 =  2 \sum_{j=1}^n \lambda_j b_j^* b_j - E_0 \I, \nonumber
\end{eqnarray}
where $E_0 = \sum_{j=1}^n \lambda_j$. Thus $H$ has been written in the form of a free Fermion system. Let $\Omega$ be the vacuum vector of the $b_j$ and
\begin{equation}\label{eq:PsiAlpha}
\psi_{\alpha} = (b_1^*)^{\alpha_1} \ldots (b_n^*)^{\alpha_n} \Omega, \quad \alpha =(\alpha_1,\alpha_2,\ldots,\alpha_n) \in \{0,1\}^n
\end{equation}
the orthonormal basis of $\mathcal{H}$ associated with $\mathcal{B}$, see Appendix~\ref{App:Bogo}.
The $\psi_{\alpha}$ form a complete set of eigenvectors for $H$ with corresponding eigenvalues $2\sum_{j:\alpha_j=1} \lambda_j - E_0$, so that the spectrum of $H$ is
\begin{equation}
\sigma(H) = \left\{ \sum_{j:\alpha_j = 1} \lambda_j - \sum_{j:\alpha_j=0} \lambda_j:\; \alpha \in \{0,1\}^n \right\}.
\end{equation}

\section{Entanglement Entropy of Eigenstates} \label{sec:EEE}

\subsection{Correlation matrices and quasi free states} \label{sec:correlation}

The goal of this section is to show that the entanglement of the eigenstates $\rho_{\alpha} = |\psi_{\alpha} \rangle \langle \psi_{\alpha} |$ of $H$ with respect to the bipartite decomposition ${\mathcal H}_1 \otimes {\mathcal H}_2$ given by (\ref{bipartite}) can be expressed in terms of {\it correlation matrices} of size $2\ell \times 2\ell$ and that these correlation matrices are restrictions of suitable spectral projections for the one-particle Hamiltonian $M$. This provides the crucial connection between properties of $M$ and properties of $H$ which will be exploited in Section~\ref{sec:arealaw} to prove Theorem~\ref{mainthm}.

The expected value of an observable $A\in\mathcal{B}(\mathcal{H})$ in the mixed state $\rho\in\mathcal{B}(\mathcal{H})$ (i.e.\ $\rho \ge 0$ with $\Tr \rho =1$) is given by
\begin{equation}
\langle A\rangle_{\rho}:= \Tr A\rho.
\end{equation}

Let $\mathcal{D}=(d_1, d_1^*, d_2,d_2^*, \ldots, d_n, d_n^*)^t$ be a general Fermonic system of $\mathcal H$ in the sense of Appendix~\ref{App:Bogo}. The correlation matrix of the state $\rho$ with respect to $\mathcal{D}$ is defined to be the $2n\times 2n$ matrix
\begin{equation} \label{cormatdef}
\Gamma^\mathcal{D}_\rho:=\langle\mathcal{D}\mathcal{D}^*\rangle_\rho,
\end{equation}
with the row $\mathcal{D}^* = (d_1^*,d_1 \ldots, d_n^*, d_n)$ and $\langle\mathcal{D}\mathcal{D}^*\rangle_\rho$ to be understood in the sense of taking
expectations of each of the operator-valued
entries of the $2n \times 2n$-matrix $\mathcal{D}\mathcal{D}^*$.

Using the CAR, cyclicity of the trace, and $\tr A^* = \overline{\tr A}$, one checks that correlation matrices are generally of the form $\Gamma= \left(\Gamma_{jk}\right)_{1\leq j,k\leq n}$ with $2\times 2$-matrix-valued matrix elements 
\begin{equation} \label{cmstructure}
\Gamma_{jk}=\frac{1}{2}\left(\delta_{j,k}\I_2+\begin{pmatrix}
                                                 X_{jk} & Y_{jk} \\
                                                 -\overline{Y_{jk}} & -\overline{X_{jk}} \\
                                               \end{pmatrix}
\right).
\end{equation}
Here $X$ and $Y$ are $n\times n$-matrices such that $X^* = X$ and $Y^* = -\overline{Y}$. Note that
\begin{equation}
P \Gamma P^t= \frac{1}{2} \left( \I_{2n} + \begin{pmatrix} X & Y  \\  -\overline{Y} & -\overline{X} \end{pmatrix} \right).
\end{equation}
In particular, correlation matrices, as defined here, are self-adjoint. We mention that in the physics literature frequently the {\it Majorana Fermions} $d_j+d_j^*$, $-i(d_j-d_j^*)$, $j=1,\ldots,n$, are used to define correlation matrices, which in this case become skew-adjoint.

\vspace{.3cm}

For $f:\Lambda \rightarrow \mathbb{C}$ define
\begin{equation}
d(f) = \sum_{j=1}^n \overline{f}_j d_j
\end{equation}
and denote its adjoint by $d^*(f) = \sum_{j} f_j d_j^*$. More generally, let
\begin{equation} \label{notation:Dfg}
D(f,g):=d(f)+d^*(g).
\end{equation}
We say that the state $\rho$ is {\it quasi free} with respect to $\mathcal{D}$ if
\begin{equation} \label{quasifree}
\left\langle\prod_{j=1}^m D_j\right\rangle_\rho = \left\{\begin{array}{ll}
                                                    0, & \hbox{if $m$ is odd;} \\
                                                    \displaystyle\sum_{k=2}^m (-1)^{k}\langle D_k D_1\rangle_\rho \left\langle\prod_{\tiny{\begin{array}{c}
                                                  j=2  \\
                                                  j\neq k
                                                \end{array}}
}^m D_j\right\rangle_\rho, & \hbox{if $m$ is even.}
                                                  \end{array}
                                                \right.
\end{equation}
Here $D_j$ is short for $D(f_j,g_j)$ and a pair $f_j, g_j: \Lambda \to \C$. Also note that here and throughout this paper operator products as on the left hand side of (\ref{quasifree}) are to be read from the right to left, i.e.\ $\prod_{j=1}^m D_j = D_m  \ldots D_1$.

For the case of even $m$, where (\ref{quasifree}) is a form of {\it Wick's rule}, an iterative application shows that expectations $\left\langle \prod_{j=1}^m D_j \right\rangle_{\rho}$ can be written as a sum of products of terms of the form $\langle D_r D_s \rangle_{\rho}$. We note that the resulting expression is known as the {\it Pfaffian} (denoted by $\pf$) of the skew-adjoint $m\times m$-matrix $D^{(\rho,m)}$ with entries $D^{(\rho,m)}_{s,r} = \langle D_r D_s \rangle_{\rho}$ for $1\leq  s<r\leq n$ and appropriately extended by antisymmetry. The Pfaffian of an odd-dimensional skew-adjoint matrix is generally set to zero, so that (\ref{quasifree}) can be restated as
\begin{equation}
\left\langle\prod_{j=1}^m D_j\right\rangle_\rho = \pf\left(D^{(\rho,m)}\right).
\end{equation}

As the expectations $\langle D_r D_s \rangle_{\rho}$ are linear combinations of elements of the correlation matrix $\Gamma_{\rho}^{\mathcal{D}}$, we conclude that $\left\langle \prod_{j=1}^m D_j \right\rangle_{\rho}$ is uniquely determined by the correlation matrix. By Lemma~\ref{lemma:FermionsGeneratesH} this determines $\langle A \rangle_{\rho}$ for {\it all} $A\in \mathcal{B}(\mathcal{H})$, which in turn characterizes $\rho$. This yields the first claim of the following Lemma:

\begin{lem} \label{lem:quasifree}
If a state $\rho$ is quasi free with respect to $\mathcal{D}$, then $\rho$ is uniquely determined by the correlation matrix $\Gamma_{\rho}^{\mathcal{D}}$. Moreover, the von Neumann entropy of $\rho$ is given by
\begin{equation} \label{redtocormat}
\mathcal{S}(\rho) = -\Tr \rho \log \rho = - \tr \Gamma_{\rho}^{\mathcal D} \log \Gamma_{\rho}^{\mathcal{D}}.
\end{equation}
\end{lem}

Note here that $\Tr$ and $\tr$ denote the traces in $\mathcal{B}(\mathcal{H})$ and $\C^{2n \times 2n}$, respectively. The second claim is restated as Proposition~\ref{thm:App:main} in Appendix~\ref{App:Bogo}, where a proof is provided.

\subsection{Correlation matrices of eigenstates} \label{coreigenstates}

Next we calculate the correlation matrix of the eigenstates $\rho_{\alpha} = |\psi_{\alpha}\rangle \langle \psi_{\alpha}|$ of $H$ given by (\ref{eq:PsiAlpha}) with respect to the Jordan-Wigner Fermionic system $\mathcal{C}$.

\begin{lem} \label{rhoCcorrelation}
Assume that $M$ has simple spectrum. Then, for each $\alpha \in \{0,1\}^n$, the correlation matrix of $\rho_\alpha$ with respect to the Jordan-Wigner Fermionic system $\mathcal{C}$ in (\ref{JWsystem}) is given by
\begin{equation} \label{eq:rhoCcor}
\Gamma_{\rho_\alpha}^{\mathcal{C}}=\chi_{\Delta_\alpha}(M),
\end{equation}
i.e.\ the spectral projection for $M$ onto the set
\begin{equation}
\Delta_\alpha:=\{\lambda_j:\alpha_j=0\}\cup\{-\lambda_j:\alpha_j=1\}.
\end{equation}
Here $\{\lambda_j,-\lambda_j: j=1,2,\ldots,n\}$ are the eigenvalues of $M$.
\end{lem}

Note that, in particular, $\rho_0=|\Omega\rangle\langle \Omega|$ is the ground state projection for $H$ and $\Gamma^{\mathcal{C}}_{\rho_0}= \chi_{\{\lambda_1,\ldots,\lambda_n\}}(\tilde{M}) = \chi_{(0,\infty)}(\tilde{M})$. The difference between $\Gamma^{\mathcal{C}}_{\rho_0}$ and $\Gamma_{\rho_\alpha}^{\mathcal{C}}$ is that $\lambda_j$ is replaced by $-\lambda_j$ for each site $j$ in which a ``particle is created'' by $b_j^*$ in (\ref{eq:PsiAlpha}).

\begin{proof}
The first step in the proof is to show that
\begin{equation}\label{eq:proof:correlation}
\Gamma^{\mathcal{B}}_{\rho_{\alpha}}=\bigoplus_{k=1}^n\begin{pmatrix}
                                       \delta_{\alpha_k,0} & 0 \\
                                       0 & \delta_{\alpha_k,1} \\
                                     \end{pmatrix},
\end{equation}
where $\mathcal{B}$ is the Fermionic system (\ref{Bsystem}).
First note that, using the definition (\ref{eq:PsiAlpha}) and the anti-commutation properties of the $b_j$, one has
\begin{equation} \label{bpsi1}
b_j^* \psi_{\alpha} = \left\{ \begin{array}{ll} 0 & \mbox{if} \;\alpha_j=1, \\ \pm\psi_{\alpha+e_j} & \mbox{if} \;\alpha_j=0, \end{array} \right.
\end{equation}
as well as
\begin{equation} \label{bpsi2}
b_j \psi_{\alpha} = \left\{ \begin{array}{ll} 0 & \mbox{if} \;\alpha_j=0, \\ \pm\psi_{\alpha-e_j} & \mbox{if} \;\alpha_j=1. \end{array} \right.
\end{equation}
Then statement (\ref{eq:proof:correlation}) follows from the following calculations
\begin{eqnarray}
\Tr b_j b_k^* \rho_{\alpha} & = & \langle \psi_{\alpha}, b_j b_k^* \psi_{\alpha} \rangle = \langle b_j^* \psi_{\alpha}, b_k^* \psi_{\alpha} \rangle \\
& = & \left\{ \begin{array}{ll} 0 & \mbox{if} \; \alpha_j = 1 \;\mbox{or} \; \alpha_k =1, \\ \langle \pm\psi_{\alpha+e_j}, \pm\psi_{\alpha+e_k} \rangle =\delta_{jk} & \mbox{if} \; \alpha_j = \alpha_k =0,  \end{array} \right. \nonumber
\end{eqnarray}
And,
\begin{eqnarray}
\Tr b_j b_k \rho_{\alpha} & = & \langle \psi_{\alpha}, b_j b_k \psi_{\alpha} \rangle = \langle b_j^* \psi_{\alpha}, b_k \psi_{\alpha} \rangle \\
& = & \left\{ \begin{array}{ll} 0 & \mbox{if} \; \alpha_j = 1 \;\mbox{or} \; \alpha_k =0, \\ \pm \langle\psi_{\alpha+e_j}, \psi_{\alpha-e_k} \rangle = 0 & \mbox{if} \;\alpha_j =0 \;\mbox{and} \;\alpha_k= 1, \end{array} \right. \nonumber
\end{eqnarray}

The general structure (\ref{cmstructure}) of correlation matrices determines the remaining entries of $\Gamma^{\mathcal{B}}_{\rho_{\alpha}}$ and completes the proof of (\ref{eq:proof:correlation}).

We now prove (\ref{eq:rhoCcor}). As $\mathcal{C} = W^t \mathcal{B}$, by Lemma~\ref{thm:Bogo:transf},
\begin{equation} \label{corchange}
\Gamma_{\rho_\alpha}^{\mathcal{C}} = W^t \Gamma_{\rho_\alpha}^{\mathcal{B}} W = W^t \bigoplus_{k=1}^n\begin{pmatrix} \delta_{\alpha_k,0} & 0 \\ 0 & \delta_{\alpha_k,1} \end{pmatrix} W.
\end{equation}

Now, since the eigenvalues of $M$ are simple, i.e.\ $-\lambda_n < \ldots < -\lambda_1 < 0 < \lambda_1 < \ldots < \lambda_n$, $\bigoplus_{k=1}^n\begin{pmatrix} \delta_{\alpha_k,0} & 0 \\ 0 & \delta_{\alpha_k,1} \end{pmatrix}$ is the spectral projection for $\bigoplus_{k=1}^n\begin{pmatrix} \lambda_k & 0 \\ 0 & -\lambda_k \end{pmatrix}$ onto $\Delta_{\alpha}$. Therefore (\ref{eq:rhoCcor}) follows from (\ref{diagM}).
\end{proof}

\begin{lem}\label{lemma:Wicks:diagonal}
The eigenstates $\rho_\alpha$ are quasi free with respect to the Jordan-Wigner Fermionic system $\mathcal{C}$.
\end{lem}
\begin{proof}
Since $\mathcal{B}$ is a Bogoliubov transformation of $\mathcal{C}$, by Lemma~\ref{lemma:BogoWicks1}(b) it suffices to prove that $\rho_\alpha$ is Wick with respect to $\mathcal{B}$. Let $\Omega_c$ and $\psi_{\alpha}^{(c)} = (c_1^*)^{\alpha_1} \ldots (c_n^*)^{\alpha_n} \Omega_c$ be the vacuum and eigenstates associated with $\mathcal{C}$. By Lemma~\ref{lemma:BogoWicks1}(a) we need to prove that $\rho_{\alpha}^{(c)} := |\psi_{\alpha}^{(c)} \rangle \langle \psi_{\alpha}^{(c)}|$ is quasi free with respect to $\mathcal{C}$ for all $\alpha$.

 This can be done by finding $\rho^{(c)}_{\alpha}$ explicitly. From the definition (\ref{eq:cjs}) of the $c_j$ we see that $c_j \begin{pmatrix} 0 \\ 1 \end{pmatrix}^{\otimes n} = 0$ for all $j$. This gives that, up to a phase, $\Omega_c = \begin{pmatrix} 0 \\ 1 \end{pmatrix}^{\otimes n}$ is the vacuum for $\mathcal{C}$, and therefore
\begin{equation}
 \rho_{0}^{(c)} = |\Omega_c\rangle \langle \Omega_c| =\bigotimes_{j=1}^n \begin{pmatrix}
                                                                                     0 & 0 \\
                                                                                     0 & 1 \\
                                                                                   \end{pmatrix}=\prod_{j=1}^n c_j c_j^*.
\end{equation}
Since for $1\leq k\leq n$ we have
\begin{equation}
c_k^* \left(\prod_{j=1}^n c_j c_j^*\right) c_k= \left(\prod_{\tiny \begin{array}{c}
                                                                     j=1 \\
                                                                     j\neq k
                                                                   \end{array}}^n c_j c_j^*\right) c_k^* c_k c_k^* c_k =\left(\prod_{\tiny \begin{array}{c}
                                                                                         j=1 \\
                                                                                         j\neq k
                                                                                       \end{array}}^n c_j c_j^*\right) c_k^* c_k,
\end{equation}
it is easy to see that
\begin{eqnarray}
\rho^{(c)}_{\alpha}&=&\prod_{j=1}^n (c_{n-j}^*)^{\alpha_{n-j}}|\Omega_c\rangle\langle\Omega_c| \prod_{j=1}^n (c_{j})^{\alpha_{j}}=\prod_{\tiny \begin{array}{c}
                                    j=1 \\
                                    \alpha_j=0
                                  \end{array}}^n c_j c_j^* \prod_{\tiny \begin{array}{c}
                                                                          j=1 \\
                                                                          \alpha_j=1
                                                                        \end{array}}^n c_j^* c_j\\
&=&\bigotimes_{j=1}^n \begin{pmatrix}
                    \delta_{\alpha_j,1} & 0 \\
                    0 & \delta_{\alpha_j,0}
                    \end{pmatrix}. \nonumber
\end{eqnarray}
Thus, as a product of diagonal states, $\rho_{\alpha}^{(c)}$ is quasi free with respect to $\mathcal{C}$ by Proposition~\ref{thm:Wicks:generalmainthm}.
\end{proof}

By Lemma~\ref{lem:quasifree} we may conclude that $\Tr \rho_{\alpha} \log \rho_{\alpha} = \tr \Gamma_{\rho_{\alpha}}^{\mathcal{C}} \log \Gamma_{\rho_{\alpha}}^{\mathcal{C}}$. This, however, is a trivial fact  as $\rho_{\alpha}$ is a pure state in $\mathcal{H}$ and $\Gamma_{\rho_{\alpha}}^{\mathcal{C}}$ an orthogonal projection in $\C^{2n}$ and thus both traces vanish. The actual use of Lemmas~\ref{rhoCcorrelation} and \ref{lemma:Wicks:diagonal} is that they provide a first step in establishing a similar trace identity for the {\it reduced} states (with non-vanishing entropy).

\subsection{Entanglement of reduced states} \label{corredstates}

Restrictions of $\Gamma_{\rho_\alpha}^{\mathcal{C}}$, which we have identified in Lemma~\ref{rhoCcorrelation} with spectral projections for $M$, can be used to express the entanglement of the states $\rho_{\alpha}$. Since $\rho_{\alpha}$ is quasi free with respect to $\mathcal{C}$ we can apply the following general theorem for such states. Here, for any subinterval $\Lambda_1 = \{r,\ldots, r+\ell-1\}$, introduce {\it local Jordan-Wigner operators} $\{c_{j}^{(1)},\ j\in\Lambda_1\}$ on $\mathcal{H}_1$ as
\begin{equation}
c_r^{(1)} := a_r, \quad c_j^{(1)} := \sigma_r^Z \ldots \sigma_{j-1}^{Z} a_j, \quad j=r+1,\ldots, r+\ell-1.
\end{equation}
The operators
\begin{equation} \label{locJW}
\mathcal{C}_1 := (c_r^{(1)}, (c_r^{(1)})^*,c_r^{(2)}, (c_r^{(2)})^*, \ldots, c_{r+\ell-1}^{(1)},(c_{r+\ell-1}^{(1)})^*)^t
\end{equation}
are a Fermionic system in $\mathcal{H}_1$.

\begin{thm} \label{thm:EErule}
If $\rho$ is quasi free with respect to $\mathcal{C}$, then the entanglement entropy of $\rho$ with respect to the bipartite decomposition ${\mathcal H}_1 \otimes {\mathcal H}_2$ is given by
\begin{equation}\label{eq:EEmain}
\mathcal{E}(\rho)=-\tr \Gamma_{\rho_1}^{\mathcal{C}_1} \log \Gamma_{\rho_1}^{\mathcal{C}_1},
\end{equation}
where $\mathcal{C}_1$ is the local Jordan-Wigner Fermionic system on $\mathcal{H}_1$. Moreover,
\begin{equation} \label{restrictGamma}
\Gamma_{\rho_1}^{\mathcal{C}_1} \;\mbox{is the restriction of} \; \Gamma_{\rho}^{\mathcal{C}} \; \mbox{to span}\,\{e_{2j-1}, e_{2j}: j\in \Lambda_1\},
\end{equation}
i.e.\ the $2\ell \times 2\ell$-submatrix of $\Gamma_{\rho}^{\mathcal{C}}$ consisting of the $2\times 2$-matrix elements corresponding to $\Lambda_1$.
\end{thm}

The proof of this theorem fills the rest of this section. Let
\begin{equation}
\tilde{c}_{j}:=\I^{\otimes (r-1)}\otimes c^{(1)}_{j}\otimes \I^{\otimes (n-\ell-r+1)}, \quad j\in \Lambda_1,
\end{equation}
be the extensions of the operators $c^{(1)}_{j}$ to $\mathcal{H}$. The $\tilde{c}_j$ are local observables for the subsystem $\Lambda_1$ and differ from the $c_j$ by products of `left-end' Pauli matrices,
\begin{equation} \label{eq:ctildec}
c_j = \sigma^Z_1 \sigma^Z_2\ldots \sigma^Z_{r-1} \tilde{c}_{j} = \tilde{c}_j \sigma^Z_1 \sigma^Z_2\ldots \sigma^Z_{r-1}, \quad j \in \Lambda_1.
\end{equation}
To describe the difference between $c_j$ and $\tilde{c}_j$ further, for $f:\Lambda_1\rightarrow \mathbb{C}$ we set
\begin{equation}
 \tilde{c}_{r}(f)=\sum_{j\in\Lambda_1} \bar{f}_j\tilde{c}_{j},
\quad c_r(f)=\sum_{j\in\Lambda_1} \overline{f}_j c_{j},
\end{equation}
and denote their adjoints by $\tilde{c}_{r}^*(f)$ and $c_r^*(f)$. For pairs $f,g:\Lambda_1\rightarrow \mathbb{C}$ we also write
\begin{equation} \label{Cops}
 \tilde{C}_r(f,g)=\tilde{c}_r(f)+\tilde{c}_r^*(g),\quad C_r(f,g)=c_r(f)+c_r^*(g).
\end{equation}

From (\ref{eq:ctildec}) and basic properties of the Pauli matrices, in particular $(\sigma_1^Z \ldots \sigma_{r-1}^Z)^2= \I$, we easily get

\begin{lem} \label{lemma:app:relations}
Let $f,g, f_1, g_1, f_2, g_2:\Lambda_1\rightarrow \mathbb{C}$, then

(a) $\tilde{C}_{r}(f,g)= \sigma^Z_1 \sigma^Z_2\ldots \sigma^Z_{r-1} C_{r}(f,g)$,

 (b) $\tilde{C}_r(f_1,g_1)\tilde{C}_{r}(f_2,g_2)=C_r(f_1,g_1)C_{r}(f_2,g_2)$
\end{lem}

With this we reach our first goal:

\begin{lem}\label{lemma:corr:RestofC}
The correlation matrix $\Gamma_{\rho_1}^{\mathcal{C}_1}$ of the reduced state $\rho_{1} =\Tr_{\mathcal{H}_2} \rho$ with respect to the local Jordan-Wigner Fermionic system $\mathcal{C}_1$ on $\mathcal{H}_1$ satisfies (\ref{restrictGamma}).
\end{lem}

\begin{proof}
That the upper right elements of the $2\times2$ matrix-elements in each correlation matrix in (\ref{restrictGamma}) coincide is seen as follows: For $j,k\in\Lambda_1$,
\begin{eqnarray}
\langle{c_{j}^{(1)}} {c_{k}^{(1)}} \rangle_{\rho_{1}} & = &\Tr \left({c_{j}^{(1)}}{c_{k}^{(1)}}\rho_{1}\right) =  \Tr \left({c_{j}^{(1)}}{c_{k}^{(1)}}\Tr_{\mathcal{H}_2}\rho\right) \\
& = & \Tr \left(\Tr_{\mathcal{H}_2}\tilde{c}_{j} \tilde{c}_{k} \rho\right)
= \Tr \left(c_{j} c_{k} \rho\right)
= \langle c_{j} c_{k} \rangle_{\rho}. \nonumber
\end{eqnarray}
Here we have used in the third step that
\begin{equation}
{c_{j}^{(1)}}{c_{k}^{(1)}} \Tr_{\mathcal{H}_2} \rho = \Tr_{\mathcal{H}_2}\left(\I\otimes {c_{j}^{(1)}} {c_{k}^{(1)}} \otimes \I\right)\rho = \Tr_{\mathcal{H}_2} \tilde{c}_j \tilde{c}_k \rho
\end{equation}
and in the fourth step that $\tilde{c}_j \tilde{c}_k = c_j c_k$ by Lemma~\ref{lemma:app:relations}(b).

Identity of the other three elements of the $2\times2$ matrix-elements follows in the same way when replacing $c_j^{(1)}$ and/or $c_k^{(1)}$ by $(c_j^{(1)})^*$ and/or $(c_k^{(1)})^*$.
\end{proof}

In order to conclude that  $\rho_{1}$ can be determined from the correlation matrix $\Gamma_{\rho_{1}}^{\mathcal{C}_1}$, we prove that is is quasi free with respect to $\mathcal{C}_1$:

\begin{lem}\label{thm:Wicks:rho1}
If $\rho$ is quasi free with respect to $\mathcal{C}$, then the reduced state $\rho_1=\Tr_{\mathcal{H}_2}\rho$ is quasi free with respect to $\mathcal{C}_1$.
\end{lem}

\begin{proof}
We need to prove (\ref{quasifree}) for $\rho_{1}$-expectations of the operators
\begin{equation}
C^{(1)}(f,g)=c^{(1)}(f)+(c^{(1)}(g))^*, \quad c^{(1)}(f)=\sum_{j\in\Lambda_1} \bar{f}_j c^{(1)}_{j}
\end{equation}
on $\mathcal{H}_1$, where $f, g: \Lambda_1 \to \C$. The latter reduce to $\rho$-expectations of the operators $\tilde{C}_r(f,g)$ defined in (\ref{Cops}): For any positive integer $m$ and functions $f_j,g_j:\Lambda_1\rightarrow\mathbb{C}$ for $1\leq j\leq m$ we have,
\begin{eqnarray} \label{Wick1}
\left\langle\prod_{j=1}^m C^{(1)}(f_j,g_j)\right\rangle_{\rho_{1}}&=&
\Tr\left( \prod_{j=1}^m C^{(1)}(f_j,g_j)\Tr_{\mathcal{H}_2}\rho\right)
\\ & = & \Tr \left(\prod_{j=1}^m \tilde{C}_r(f_j,g_j)\rho\right)
= \left\langle\prod_{j=1}^m \tilde{C}_r(f_j,g_j)\right\rangle_{\rho}, \nonumber
\end{eqnarray}
which uses the general fact $\Tr_{\mathcal{H}_2} (B \otimes I)A = B \Tr_{\mathcal{H}_2} A$ for $A\in \mathcal{B}(\mathcal{H})$ and $B\in \mathcal{B}(\mathcal{H}_1)$ in the second step.

{\bf Case 1:} If $m$ is even, then using Lemma \ref{lemma:app:relations} we get that
\begin{equation}\label{eq:app:evenproductofc1s}
\prod_{j=1}^m \tilde{C}_r(f_j,g_j)=\prod_{j=1}^m C_r(f_j,g_j).
\end{equation}
Note that the operators $C_r(f_j,g_j)$ are of the form $C(\tilde{f}_j, \tilde{g}_j)$ where $\tilde{f}_j$, $\tilde{g}_j$ are the extensions of $f_j$, $g_j$ by zeros to $\Lambda$. Thus we can use that $\rho$ is quasi free with respect to $\mathcal{C}$, i.e.\ Lemma \ref{lemma:Wicks:diagonal}:
\begin{equation} \label{Wick2}
\left\langle\prod_{j=1}^m C^{(1)}(f_j,g_j)\right\rangle_{\rho_{1}}=
\left\langle\prod_{j=1}^m C_r(f_j,g_j)\right\rangle_{\rho}=\pf\left(C_r^{(\rho,m)}\right),
\end{equation}
where $C_r^{(\rho,m)}$ is the $m\times m$ anti-symmetric matrix with entries
\begin{equation}
\left[C_r^{(\rho,m)}\right]_{j,k}=\langle C_r(f_k,g_k)C_r(f_j,g_j)\rangle_{\rho}
\end{equation}
for $1\leq j<k\leq n$ and properly extended by antisymmetry.
Arguing as in (\ref{Wick1}) and (\ref{eq:app:evenproductofc1s}), but now with only two factors,
\begin{equation}
\langle C_r(f_j,g_j)C_r(f_k,g_k)\rangle_{\rho}=\langle \widetilde{C}_r(f_j,g_j)\widetilde{C}_r(f_k,g_k)\rangle_{\rho}=\langle C^{(1)}(f_j,g_j)C^{(1)}(f_k,g_k)\rangle_{\rho_{1}}.
\end{equation}
Combined with (\ref{Wick2}) this shows that Wick's rule is satisfied if $m$ is even.

{\bf Case 2:} If $m$ is odd, then by (\ref{Wick1}), (\ref{eq:app:evenproductofc1s}) and Lemma~\ref{lemma:app:relations} we get
\begin{eqnarray} \label{oddcalc}
\left\langle\prod_{j=1}^m C^{(1)}(f_j,g_j)\right\rangle_{\rho_{1}}&=&
\left\langle\tilde{C}_r(f_m,g_m)\prod_{j=1}^{m-1} \tilde{C}_r(f_j,g_j)\right\rangle_{\rho} \\
&=& \left\langle\tilde{C}_r(f_m,g_m)\prod_{j=1}^{m-1} C_r(f_j,g_j)\right\rangle_{\rho} \nonumber \\
&=&\left\langle \left(\sigma_1^Z \sigma_2^Z\ldots\sigma_{r-1}^Z\right)C_r(f_m,g_m)\prod_{j=1}^{m-1} C_r(f_j,g_j)\right\rangle_{\rho} \nonumber \\
&=& \left\langle \left(\sigma_1^Z \sigma_2^Z\ldots\sigma_{r-1}^Z\right)\prod_{j=1}^{m} C_r(f_j,g_j)\right\rangle_{\rho}. \nonumber
\end{eqnarray}
Now note that
\begin{equation}
\sigma_1^Z \sigma_2^Z\ldots\sigma_{r-1}^Z=\prod_{j=1}^{r-1}(c_j+c_j^*)(c_j-c_j^*).
\end{equation}
Thus $\sigma_1^Z \sigma_2^Z\ldots\sigma_{r-1}^Z$ is an even product of operators $C(\hat{f}_j,\hat{g}_j)$ with suitable choices of $\hat{f}_j$ and $\hat{g}_j$. Since $\rho$ is quasi free with respect to $\mathcal{C}$, we conclude that the right hand side of (\ref{oddcalc}) is zero.
\end{proof}

\vspace{.3cm}

Combining Lemma~\ref{thm:Wicks:rho1} and Lemma~\ref{lem:quasifree} gives
\begin{equation}
\mathcal{E}(\rho) = \mathcal{S}(\rho_{1}) = - \tr \Gamma_{\rho_{1}}^{\mathcal{C}_1} \log \Gamma_{\rho_{1}}^{\mathcal{C}_1},
\end{equation}
which completes the proof Theorem~\ref{thm:EErule}.

\section{An Area Law for the Eigenstates} \label{sec:arealaw}

In this section we {\it prove} Theorem~\ref{mainthm}, our main result.

Note first that, by the non-degeneracy assumption (\ref{simple}), it holds with probability  one that all eigenprojectors $\rho_{\psi}$ of $H$ are of the form $\rho_{\alpha} = |\psi_{\alpha}\rangle \langle \psi_{\alpha}|$. Therefore (\ref{arealaw}) is equivalent to
\begin{equation} \label{arealaw2}
\mathbb{E} \left( \max_{\alpha} \mathcal{E}(\rho_{\alpha}) \right) \le \tilde{C} < \infty
\end{equation}
uniformly in $n$, $r$ and $\ell$, meaning we can work with the Fermion basis given by (\ref{eq:PsiAlpha}).

For fixed $\alpha$, set $\Gamma_1 := \Gamma_{(\rho_\alpha)_1}^{\mathcal{C}_1}$ the correlation matrix of $(\rho_{\alpha})_1:=\Tr_{\mathcal{H}_2}\rho_\alpha$ with respect to $\mathcal{C}_1$. $\Gamma_1$ is the restriction (\ref{restrictGamma}) of $\Gamma := \Gamma_{\rho_{\alpha}}^{\mathcal{C}} = \chi_{\Delta_{\alpha}}(M)$. Let the eigenvalues of $\Gamma_1$ be $\xi_j$, $1-\xi_j$, $j=1,\ldots, \ell$, with $0 \le \xi_j \le 1/2$. 

The following is a calculation essentially taken from \cite{PasturSlavin}, extended to the more general type of correlation matrices needed here to cover quadratic Fermion Hamiltonians of the form (\ref{eq:Crep}).

By Theorem~\ref{thm:EErule} we have
\begin{eqnarray} \label{longcalc1}
\mathcal{E}(\rho_{\alpha}) & = & - \tr \Gamma_1 \log \Gamma_1 \\
& = & - \sum_{j=1}^\ell \left( \xi_j \log \xi_j + (1-\xi_j) \log (1-\xi_j) \right) \nonumber \\
& \le & 2 \log 2 \sum_{j=1}^\ell \sqrt{\xi_j (1-\xi_j)} \nonumber \\
& = & \log 2 \tr (\Gamma_1(\I -\Gamma_1))^{1/2}, \nonumber
\end{eqnarray}
where we have used the elementary inequality
\begin{equation}
-x\log x-(1-x)\log (1-x)\leq 2\log 2 \sqrt{x(1-x)},\ \ \text{for}\ \  0<x<1.
\end{equation}
The Peierls-Bogoliubov inequality, see Section 8.3 in  \cite{SimonTrace}, says that
\begin{equation}
\tr f(A) \geq \sum_{j=1}^m f(A_{jj})
\end{equation}
for any convex function $f$ and $m\times m$ hermitian matrix $A$. Using this with $f(x)=-\sqrt{x}$ as well as the elementary inequality $\sqrt{x}+\sqrt{y}\leq \sqrt{2}\sqrt{x+y}$, we may further bound (\ref{longcalc1}) by
\begin{equation}\label{longcalc3}
\mathcal{E}(\rho_{\alpha}) \leq \sqrt{2} \log 2 \sum_{j\in \Lambda_1} \left( \tr (\Gamma_1 (\I - \Gamma_1))_{jj} \right)^{1/2},
\end{equation}
where matrix elements should be understood as $2\times 2$-matrices.

Now, since $\Gamma$ is an orthogonal projection, we use $\Gamma^2=\Gamma$ with block matrix multiplication to get,
\begin{equation}
\Gamma_{jj}=\Gamma_{jj}^2+\sum_{\tiny \begin{array}{c}
                                     k\in\Lambda_1 \\
                                     j\neq k
                                   \end{array}} \Gamma_{jk}\Gamma_{kj}+
                                   \sum_{k\in\Lambda\backslash\Lambda_1}
                                   \Gamma_{jk}\Gamma_{kj}.
\end{equation}
Then, for $j\in\Lambda_1$,
\begin{equation}\label{eq:arealaw:pf:2}
\left(\Gamma_1(\I-\Gamma_1)\right)_{jj}
=\Gamma_{jj}(\I-\Gamma_{jj})-\sum_{\tiny\begin{array}{c}
             k \in\Lambda_1 \\
             k \neq j
           \end{array}}
\Gamma_{jk}\Gamma_{kj}
=\sum_{k\in\Lambda \setminus \Lambda_{1}}\Gamma_{jk}
\Gamma_{kj} = \sum_{k\in\Lambda \setminus \Lambda_{1}} \Gamma_{jk}
(\Gamma_{jk})^t,
\end{equation}
where symmetry of $\Gamma$ was used (note that $\Gamma_1$ is real-valued, see (\ref{corchange}), where $W$ is real). Inserting this into (\ref{longcalc3}), using $\sqrt{x+y} \le \sqrt{x}+\sqrt{y}$ as well as $\tr \Gamma_{jk} (\Gamma_{jk})^t \le 2\|\Gamma_{jk}\|^2$, we find
\begin{equation}
\mathcal{E}(\rho_\alpha)\leq  2\log 2 \sum_{\textbf{j}\in\Lambda_1}\sum_{\textbf{k}\in\Lambda\backslash\Lambda_1}
\|\Gamma_{jk}\|.
\end{equation}
Maximizing over $\alpha$ and averaging gives
\begin{equation} \label{boundbyproj}
\mathbb{E}\left(\sup_\alpha\mathcal{E}(\rho_\alpha)\right)\leq  2\log 2 \sum_{j \in\Lambda_1}\sum_{k\in\Lambda\backslash\Lambda_1}
\mathbb{E}\left(\sup_\alpha\left\|\left[\chi_{\Delta_\alpha}(M)\right]_{jk}\right\|\right)
\end{equation}
Now, by assumption (\ref{ecloc}),
\begin{equation}
\mathbb{E}\left(\sup_\alpha\left\|\left[\chi_{\Delta_\alpha}(M)\right]_{jk}\right\|\right)\leq
\mathbb{E}\left(\sup_{|g|\leq 1}\left\|\left[g(M)\right]_{jk}\right\|\right)\leq \frac{C}{1+|j-k|^\beta}
\end{equation}
for some $\beta>2$. We have
\begin{eqnarray} \label{eq:proof:AreaLaw:last}
\sum_{j\in\Lambda_1}\sum_{k\in\Lambda\backslash\Lambda_1}\frac{1}{1+|j-k|^\beta} & \le & \sum_{j=1}^{\ell} \sum_{k \in \Z \setminus \{1,\ldots,\ell\}} \frac{1}{1+|j-k|^{\beta}} \\
& = & 2 \sum_{j=1}^{\ell} \sum_{k=\ell+1}^{\infty} \frac{1}{1+(k-j)^{\beta}} \nonumber \\
& \le & 2 \sum_{j=1}^{\ell} \sum_{k=\ell+1}^{\infty} \frac{1}{\sqrt{1+2(\ell-j)^{\beta}}} \frac{1}{\sqrt{1+2(k-\ell)^{\beta}}} \nonumber \\
& \le & 2 \left( \sum_{j=0}^{\infty} \frac{1}{\sqrt{1+2j^{\beta}}} \right)^2 < \infty. \nonumber
\end{eqnarray}
This gives the uniform boundedness of (\ref{boundbyproj}) in $n$, $r$ and $\ell$, and thus completes the proof of Theorem~\ref{mainthm}.

\appendix

\section[Non-Degeneracy of Eigenvalues]{Non-Degeneracy of Eigenvalues}\label{App:Non-degeneracy}

Here we will prove that the non-degeneracy assumption (\ref{simple}) holds under the condition that the random variables $\nu_j$, $j\in \N$, are i.i.d.\ with absolutely continuous distribution and that $\mu_j$, $\gamma_j$, $j\in \N$, are independent of the $\nu_j$. In particular, this covers the applications discussed in Section~\ref{sec:Applications}. In this case, almost sure non-degeneracy of the eigenvalues of $H$ will follow from Proposition~\ref{prop:nondeg} below.

Towards this, let $A$ is a hermitian $n\times n$-matrix, $B$ an anti-hermitian $n\times n$-matrix, $A(\nu) = A + \mbox{diag}(\nu)$ for $\nu = (\nu_1, \ldots, \nu_n) \in \R^n$ and
\begin{equation}
\tilde{M}(\nu) = \left( \begin{array}{cc} A(\nu) & B \\ -B & -A(\nu) \end{array} \right).
\end{equation}
As $\tilde{M}(\nu)$ and $-\tilde{M}(\nu)$ are unitarily equivalent, $\tilde{M}(\nu)$ has $n$ pairs $\pm \lambda_j$, $j=1,\ldots,n$, of eigenvalues, where we may choose $0 \le \lambda_1 \le \ldots \le \lambda_n$. As discussed in Section~\ref{sec:JordanWigner}, the $2^n$ eigenvalues of the corresponding many-body Hamiltonian $H(\nu) = {\mathcal C}^* P^t \tilde{M}(\nu) P {\mathcal C}$ are
\begin{equation}
E_{\alpha} := \sum_{j:\alpha_j = 1} \lambda_j - \sum_{j:\alpha_j=0} \lambda_j
\end{equation}
for multi-indices $\alpha = (\alpha_1, \ldots, \alpha_n) \in \{0,1\}^n$.

\begin{proposition} \label{prop:nondeg}
For Lebesgue-almost every $\nu = (\nu_1, \ldots, \nu_n) \in \R^n$ the $2^n$ eigenvalues $\{E_{\alpha}: \alpha \in \{0,1\}^n\}$ of $H(\nu)$ are pairwise distinct.
\end{proposition}

\begin{remark} Pairwise distinctness of the many-body eigenvalues $E_{\alpha}$ readily implies pairwise distinctness of the one-body eigenvalues, i.e.\ Proposition~\ref{prop:nondeg} yields that $0 < \lambda_1 < \ldots < \lambda_n$ and, in particular, invertibility of $\tilde{M}(\nu)$ for Lebesgue-almost every $\nu \in \R^n$.
\end{remark}

The elementary proof of Proposition~\ref{prop:nondeg} proceeds in two steps: First we show the existence of {\it one} $\nu' \in \R^n$ such that the $2^n$ eigenvalues of $H(\nu')$ are pairwise distinct. Then we use analytic perturbation theory to show that this property extends to $H(\nu)$ for Lebesgue-a.e.\ $\nu \in \R^n$.

\vspace{.3cm}

{\bf Step 1:} Consider the functions $f_{\beta}(x) := \sum_{j=1}^n \beta_j x_j$ on $\R^n$, $\beta \in \{-1,1\}^n$. There exists a nullset $N \subset \R^n$ such that the numbers $f_{\beta}(x)$, $\beta \in \{-1,1\}^n$, are pairwise distinct for all $x\in \R^n \setminus N$ (first observe that for each pair $\beta \not= \tilde{\beta}$ there is a nullset $N_{\beta \tilde{\beta}}$ such that $f_{\beta}(x) \not= f_{\tilde{\beta}}(x)$ for all $x\in \R^n \setminus N_{\beta \tilde{\beta}}$, then take the union of these nullsets).

Fix $\nu'' \in \R^n \setminus N$ with $0 < \nu_1'' < \ldots < \nu_n''$ and let $\delta := \min_{\beta \not= \tilde{\beta}} |f_{\beta}(\nu'')- f_{\tilde{\beta}}(\nu'')|$. Let $D := \left\| \begin{pmatrix} A & B \\ -B & -A \end{pmatrix} \right\|$ and $\nu' := \frac{4nD+1}{\delta} \nu''$, so that
\begin{equation}
\min_{\beta \not= \tilde{\beta}} |f_{\beta}(\nu') - f_{\tilde{\beta}}(\nu')| = 4nD+1.
\end{equation}

Suitable choices of $\beta$, $\tilde{\beta}$ show that $|\nu_k'|>2D$ and $|\nu_k'-\nu_{\ell}'| > 2D$ for all $k \not= \ell$. Treating $M(\nu')$ as a perturbation of diag$(\nu_1', \ldots, \nu_n',-\nu_1', \ldots, -\nu_n')$, we see that each of the $2n$ intervals
\begin{equation}
[\nu_j'-D, \nu_j'+D], \: [-\nu_j'-D, -\nu_j'+D], \quad j=1,\ldots,n,
\end{equation}
contains exactly one eigenvalue of $M(\nu')$. Denote these eigenvalues by $\pm\lambda_j'$, $j=1,\ldots,n$. For the corresponding eigenvalues $E_{\beta}' := f_{\beta}(\lambda_1',\ldots,\lambda_n')$, $\beta \in \{ 1,-1\}^n$, of $H(\nu')$ we get
\begin{eqnarray}
|E_{\beta}' - E_{\tilde{\beta}}'| & = & \left| \sum_{j=1}^n (\beta_j - \tilde{\beta}_j) \lambda_j' \right| \\
& = & \left| \sum_{j=1}^n (\beta_j - \tilde{\beta}_j) (\lambda_j'-\nu_j') + f_{\beta}(\nu') - f_{\tilde{\beta}}(\nu') \right| \nonumber \\
& \ge & |f_{\beta}(\nu') - f_{\tilde{\beta}}(\nu') | - 2 \sum_{j=1}^n |\lambda_j' - \nu_j'| \nonumber \\
& \ge & |f_{\beta}(\nu') - f_{\tilde{\beta}}(\nu') | - 2nD \; > \; 0 \nonumber
\end{eqnarray}
if $\beta \not= \tilde{\beta}$. Thus $H(\nu')$ has non-degenerate spectrum.

\vspace{.3cm}

{\bf Step 2:} We now apply analytic perturbation theory iteratively to each of the parameters $\nu_1$, \ldots, $\nu_n$ in $H(\nu)$.

Fix $\nu_2'$, \ldots, $\nu_n'$, the last $n-1$ components of $\nu'$ found above. Then
\begin{equation}
H(\nu_1, \nu_2', \ldots, \nu_n') = A_1 + \nu_1 B_1
\end{equation}
for selfadjoint $A_1$ and $B_1$. Thus by Theorem II.6.1 of \cite{Kato} there are real analytic functions $f_k$, $k=1,\ldots,2^n$, such that $\{f_k(\nu_1): k = 1,\ldots, 2^n\}$ are the eigenvalues of $H(\nu_1, \nu_2', \ldots, \nu_n')$ for each $\nu_1 \in \R$. As $H(\nu')$ has non-degenerate spectrum, the numbers $\{f_k(\nu_1'): k=1,\ldots,2^n\}$ are pairwise distinct. Analyticity of the functions $f_k$ implies the existence of a nullset $N^{(1)} \subset \R$ (in fact a countable set) such that the numbers $\{f_k(\nu_1): k=1,\ldots, 2^n\}$ are pairwise distinct for each $\nu_1 \in \R \setminus N^{(1)}$ (each pair of functions $f_k$ and $f_{\tilde{k}}$, $k \not= \tilde{k}$, coincides at no more than countable many points).

For any such $\nu_1$, we can now use analyticity in $\nu_2$ (given that $H(\nu_1, \nu_2, \nu_3', \ldots, \nu_n') = A_2 + \nu_2 B_2$) to get the existence of a nullset $N^{(2)}(\nu_1)$ in $\R$ such that $H(\nu_1, \nu_2, \nu_3', \ldots, \nu_n')$ has pairwise distinct eigenvalues for all $\nu_2 \in \R \setminus N^{(2)}(\nu_1)$. It follows that the eigenvalues of $H(\nu_1, \nu_2, \nu_3', \ldots, \nu_n')$ are non-degenerate for all $\nu_1$, $\nu_2$ with $\nu_1 \in \R\setminus N^{(1)}$, $\nu_2 \in \R \setminus N^{(2)}(\nu_1)$, and thus for Lebesgue-a.e.\ $(\nu_1, \nu_2) \in \R^2$.

Iteration of this argument leads to Proposition~\ref{prop:nondeg}.

\section[Finite Fermionic Systems and Quasi-Free States]{Finite Fermionic Systems and Bogoliubov Transformations} \label{App:Bogo}

In this appendix we collect some further background from the theory of finite Fermionic systems, which was used in Section~\ref{sec:EEE} above. Most of this is well known in theoretical physics and can also be found in the mathematical physics literature, e.g.\ \cite{BratRob} or \cite{Bachetal}. We include a self-contained presentation of this material here, partly for the convenience of the reader, but also to ensure that these tools are available in the generality required here.

In a Hilbert space $\mathcal{H}$ of dimension dim$\,\mathcal{H} = 2^n$, we call
\begin{equation}
\mathcal{D} = (d_1,d_1^*,d_2,d_2^*, \ldots, d_n, d_n^*)^t
\end{equation}
a {\it Fermionic system} if the operators $d_j \in \mathcal{B}(\mathcal{H})$ and their adjoints satisfy the canonical commutation relations (CAR)
\begin{equation} \label{CARpropertiesD}
\{d_j, d_k^*\} = \delta_{jk}\I, \quad \{d_j, d_k\} = \{d_j^*, d_k^*\} = 0 \quad \mbox{for all $j, k =
1,\ldots,n$}.
\end{equation}
The intersection of the kernels of the $d_j$ is one-dimensional, i.e.\ they contain an essentially unique normalized vector $\Omega$, referred to as the vacuum vector,  from which an orthonormal basis of $\mathcal{H}$ is found as
\begin{equation}\label{eq:PhiAlpha}
\phi_{\alpha} = (d_1^*)^{\alpha_1} \ldots (d_n^*)^{\alpha_n} \Omega, \quad \alpha \in \{0,1\}^n,
\end{equation}
see for example \cite{Simon}.

It is easy to see that the operators $\{d_1,d_1^*, \ldots, d_n, d_n^*\}$ in a Fermionic system are linearly independent (consider a general linear combination of these operators and calculate its anti-commutators with all $d_j$ and $d_j^*$). Next we state two other basic properties of finite Fermionic systems.

\begin{lem}\label{lemma:Ferm:UnitEquiv}
Let $\mathcal{H}$ and $\tilde{\mathcal{H}}$ be $2^n$-dimensional Hilbert spaces and $\mathcal{D} = (d_1, d_1^*, \ldots,d_n,d_n^*)^t$ be a Fermionic system in $\mathcal{H}$. Then
\begin{equation}
\tilde{\mathcal{D}} = (\tilde{d}_1,\tilde{d}_1^*,\ldots,\tilde{d}_n, \tilde{d}_n^*)^t
\end{equation}
is a Fermionic system in $\tilde{\mathcal{H}}$ if and only if there exists a unitary operator $U: \mathcal{H} \to \tilde{\mathcal{H}}$ such
that $U^* \tilde{d}_j U=d_j$ for all $1\leq j\leq n$. The unitary operator $U$ is characterized by $U \phi_{\alpha} = \tilde{\phi}_{\alpha}$ for all $\alpha \in \{0,1\}^n$, where $\{\phi_{\alpha}\}$ and $\{\tilde{\phi}_{\alpha}\}$ are the ONBs associated with $\mathcal{D}$ and $\tilde{\mathcal{D}}$ through (\ref{eq:PhiAlpha}).
\end{lem}

\begin{proof}
For $U$ as above and any $j$ we have
\begin{eqnarray}
{\tilde{d}}_j^*U\phi_\alpha & = & {\tilde{d}}_j^*\tilde{\phi}_\alpha=\left\{
                              \begin{array}{ll}
                                0, & \hbox{if $\alpha_j=1$;} \\
                                (-1)^{\sum_{k=1}^j\alpha_k}\tilde{\phi}_{\alpha+e_j} = (-1)^{\sum_{k=1}^j\alpha_k} U \phi_{\alpha+e_j}, & \hbox{if $      \alpha_j=0$.} \end{array} \right. \\
& = & U d_j^* \phi_{\alpha} . \nonumber
\end{eqnarray}
Thus $d_j^*=U^* {\tilde{d}}_j^* U$ and $U^* {\tilde{d}}_j U=d_j$ for all $1\leq j\leq n$.

The converse is straightforward. \end{proof}

\begin{lem}\label{lemma:FermionsGeneratesH}
Let $\mathcal{D}$ be a Fermionic system in $\mathcal{H}$ and $\mathcal{A}$ be
the $\star$-algebra generated by the components of $\mathcal{D}$. Then $\mathcal{A}=\mathcal{B}(\mathcal{H})$.
\end{lem}

\begin{proof}
With $\phi_{\alpha}$ as in (\ref{eq:PhiAlpha}) it suffices to show that for every pair $\alpha, \beta \in \{0,1\}^n$ there exists an operator $A_{\alpha,\beta}\in\mathcal{A}$, such that $A_{\alpha,\beta}\phi_{\alpha}=\phi_{\beta}$ and $A_{\alpha,\beta}\phi_{\tilde{\alpha}}=0$ for $\tilde{\alpha}\neq \alpha$. This operator is explicitly given by
\begin{equation}
A_{\alpha,\beta}:=\left(\prod_{j=1}^{n}(d_jd_j^*)^{1-\beta_j}\right)\left(\prod_{j=0}^{n-1} \left(d_{n-j}^*\right)^{\beta_{n-j}}\right)\left(\prod_{j=1}^n d_{j}^{\alpha_{j}}\right).
\end{equation}
We omit the somewhat tedious calculations needed to verify this.
\end{proof}

A matrix $W\in \C^{2n\times 2n}$ is called a {\it Bogoliubov matrix} if $W$ is unitary and
\begin{equation} \label{eq:Bogo:WJ}
WJ W^t=J,\ \text{where }\
J=\left(\sigma^X\right)^{\oplus n}
\end{equation}

The reason for using this terminology is that, for the finite Fermionic systems considered here, Bogoliubov matrices implement Bogoliubov transformations:

\begin{lem} \label{thm:Bogo:transf}
Let $\mathcal{D}$ be a Fermionic system in $\mathcal{H}$ and $W\in \C^{2n \times 2n}$. Then
\begin{equation}\label{eq:Bogo:transf}
\tilde{\mathcal{D}}:=W\mathcal{D}
\end{equation}
is a Fermionic system in $\mathcal{H}$ if and only if $W$ is a Bogoliubov matrix. In this case, the correlation matrices, as defined in (\ref{cormatdef}) above, are related by
\begin{equation} \label{cormattransfer}
\Gamma_{\rho}^{\tilde{\mathcal{D}}} = W \Gamma_{\rho}^{\mathcal{D}} W^*
\end{equation}
for all states $\rho$ in $\mathcal{H}$.
\end{lem}

\begin{proof}
Note first that (\ref{cormattransfer}) follows from a simple linearity argument.
Let $W$ be a Bogoliubov matrix and $\mathcal{D}$ a Fermionic system. The latter means
\begin{equation}
\mathcal{D}\mathcal{D}^*+J\left(\mathcal{D}\mathcal{D}^*\right)^tJ=\I_{2n}.
\end{equation}
Note that, given unitarity, the condition (\ref{eq:Bogo:WJ}) is equivalent to $J\overline{W}=WJ$ and $W^t J=J W^*$. Thus
\begin{equation}
\tilde{\mathcal{D}}\tilde{\mathcal{D}}^* +J (\tilde{\mathcal{D}}\tilde{\mathcal{D}}^*)^t J
 =  W\left(\mathcal{D}\mathcal{D}^* +J(\mathcal{D}\mathcal{D}^*)^t J\right)W^*
 =  W\I_{2n} W^*
 =  \I_{2n},
\end{equation}
so $\tilde{\mathcal{D}}$ is a Fermionic system.

By performing a simple change of basis, the converse can be restated as: If $\mathcal{D}$ and $\tilde{\mathcal{D}}$ are two Fermionic systems related by
\begin{equation}\label{eq:Bogo:transf2}
 \left(\tilde{d}_1 , \tilde{d}_2,\ldots,\tilde{d}_n, (\tilde{d}_1)^* , \ldots , (\tilde{d}_n)^*\right)^t
=\hat{W} \left(d_1 , d_2,\ldots,d_n, (d_1)^* , \ldots , (d_n)^*\right)^t
\end{equation}
then $\hat{W}$ is unitary and satisfies 
\begin{equation}\label{eq:Bogo:2}
\hat{W}J_1\hat{W}^t=J_1, \ \text{where}\ J_1:=\begin{pmatrix}
                                                            0 & \I \\
                                                            \I & 0 \\
                                                          \end{pmatrix}.
\end{equation}
To show this, write $\hat{W}$ in block form,
\begin{equation}
\hat{W}=\begin{pmatrix}
    K & L \\
    M & N \\
  \end{pmatrix}.
\end{equation}
From the transformation (\ref{eq:Bogo:transf2}) we get that for $j=1,2,\ldots,n$,
\begin{equation}\label{eq:Bogo:bj}
\tilde{d}_j   = \sum_{k=1}^n (K_{jk} d_k+L_{jk} d_k^*), \quad \tilde{d}_j^* = \sum_{k=1}^n (M_{jk} d_k+N_{jk} d_k^*).
\end{equation}
By taking the adjoint of the left hand side of (\ref{eq:Bogo:bj}) and comparing it with the right hand side, using that the elements of a Fermionic system are linearly independent, we get
\begin{equation}
N_{jk}= \overline{K_{jk}}, \quad M_{jk}=\overline{L_{jk}}, \text{ for all } j, k=1,2,\ldots,n,
\end{equation}
meaning that
\begin{equation}\label{eq:Bogo:WBlocks}
M=\overline{L} \text{  and  } N=\overline{K}.
\end{equation}
Next, we will prove that $\hat{W}$ is a unitary matrix, that is the rows are an ONB of $\C^{2n}$.
Using (\ref{eq:Bogo:bj}) and (\ref{eq:Bogo:WBlocks}), a calculation starting with $\{\tilde{d}_j,\tilde{d}_j^*\}=\I$ proves that the rows of $\hat{W}$ are unit vectors. Next, for $j\neq k$, two calculations starting from $\{\tilde{d}_j,\tilde{d}_k^*\}=0$ and $\{\tilde{d}_j,\tilde{d}_k\}=0$ show that any two rows of $\hat{W}$ are orthogonal.
Thus $\hat{W}$ is a unitary matrix, which means that $WW^*=\I$. In terms of the blocks of $\hat{W}$ this means that
\begin{equation}
\begin{array}{l}
      KK^*+LL^*=\I, \quad  KL^t+L K^t=0, \\
      \overline{L}K^*+\overline{K}L^*=0, \quad   \overline{L}L^t+\overline{K}K^t=\I.
\end{array}
\end{equation}
One checks that this is equivalent to (\ref{eq:Bogo:2}).
\end{proof}

In the situation above we say that $\tilde{\mathcal{D}}$ is a Bogoliubov transformation of $\mathcal{D}$. One can easily see that being related by a Bogoliubov transformation is an equivalence relation between Fermionic systems.

As in Section~\ref{sec:correlation} we say that a state $\rho$ on $\mathcal{H}$ is quasi free with respect to a Fermionic system $\mathcal{D}$ if expectations of products of operators $D(f_j,g_j)$ satisfy (\ref{quasifree}).

\begin{lem}\label{lemma:BogoWicks1}
Let $\mathcal{D}$ and $\tilde{\mathcal{D}}$ be Fermionic systems on $\mathcal{H}$ and assume that $\rho$ is quasi free with respect to $\mathcal{D}$.

(a) Let $U$ be the unitary relating $\mathcal{D}$ and $\tilde{\mathcal{D}}$ as in Lemma~\ref{lemma:Ferm:UnitEquiv}. Then $U\rho U^*$ is quasi free with respect to $\tilde{\mathcal{D}}$.

(b) If $\tilde{\mathcal{D}}$ is a Bogoliubov transformation of $\mathcal{D}$, then $\rho$ is quasi free with respect to $\tilde{\mathcal{D}}$.
\end{lem}

\begin{proof}
(a) For any positive integer $m$,
\begin{equation}\label{equ:Wick:local1}
\left\langle\prod_{j=1}^m \tilde{D}_j\right\rangle_{U\rho U^*}=
\left\langle U^*\left(\prod_{j=1}^m \tilde{D}_j\right) U\right\rangle_{\rho}=
\left\langle\prod_{j=1}^m D_j\right\rangle_{\rho}.
\end{equation}
From this it is straightforward to check that $U\rho U^*$ is quasi free with respect to $\tilde{\mathcal{D}}$.

(b) Since the $\tilde{d}_j$ and $\tilde{d}_j^*$ are linear combinations of the $d_j$ and $d_j^*$, one can see that, for given $f,g: \{1,\ldots,n\} \to \C$,
\begin{equation}
\tilde{D}(f,g)=D(h,r)
\end{equation}
for some $h,r :\{1,2,\ldots,n\}\rightarrow \C$. With this the claim follows easily.
\end{proof}

The following result states that, up to unitary equivalence, quasi free states are determined by their correlation matrices.

\begin{lem}\label{thm:App:corrBogo:1}
Let $\rho$ and $\tilde{\rho}$ be states in $2^n$-dimensional Hilbert spaces $\mathcal{H}$ and $\tilde{\mathcal{H}}$, respectively, which are quasi free with respect to the Fermionic systems $\mathcal{D}$ and $\tilde{\mathcal{D}}$, respectively. If
\begin{equation}
\Gamma_{\rho}^{\mathcal{D}}=\Gamma_{\tilde{\rho}}^{\tilde{\mathcal{D}}},
\end{equation}
then $\rho$ and $\tilde{\rho}$ are unitary equivalent.
\end{lem}

\begin{proof}
Using notations as in Section~\ref{sec:correlation} (for both $\mathcal{D}$ and $\tilde{\mathcal{D}}$), equality of the correlation matrices gives the equality of the correlations
\begin{equation}
\langle D_j D_k\rangle_{\rho}=\langle \tilde{D}_j \tilde{D}_k\rangle_{\tilde{\rho}} \ \ \text{for all}\; \;1\leq j<k\leq n.
\end{equation}
Thus for any positive integer $m$,
\begin{equation}
D^{(\rho,m)}=\tilde{D}^{(\tilde{\rho},m)}.
\end{equation}
As $\rho$ and $\tilde{\rho}$ are quasi free with respect to $\mathcal{D}$ and $\tilde{\mathcal{D}}$, this implies
\begin{eqnarray}
\left\langle \prod_{j=1}^m D_j\right\rangle_{\rho}&=&\pf\left(D^{(\rho,m)}\right)=\pf\left(\tilde{D}^{(\tilde{\rho},m)}\right)
= \left\langle \prod_{j=1}^m \tilde{D}_j\right\rangle_{\tilde{\rho}}=\left\langle\prod_{j=1}^m U^* D_j U\right\rangle_{\tilde{\rho}} \\
&=&\left\langle\prod_{j=1}^m D_j\right\rangle_{U \tilde{\rho} U^*}, \nonumber
\end{eqnarray}
where we used Lemma~\ref{lemma:Ferm:UnitEquiv} which provides a unitary $U$ such that $\tilde{d}_j=U^* d_j U$ for all $1\leq j\leq n$. By Lemma~\ref{lemma:FermionsGeneratesH} this implies $\langle A \rangle_{\rho} = \langle A \rangle_{U \tilde{\rho} U^*}$ and thus
\begin{equation}
\rho=U \tilde{\rho} U^*
\end{equation}
\end{proof}

Our final goal in this Appendix will be to prove the trace identity (\ref{redtocormat}). We will first prove this identity for diagonal product states with respect to the Jordan-Wigner Fermionic operators $\mathcal{C}$ given in (\ref{JWsystem}). The general result (\ref{redtocormat}) will then be reduced to this special case.

\begin{proposition} \label{thm:Wicks:generalmainthm}
Let $\rho^{(diag)}\in\mathcal{B}(\bigotimes_{j=1}^n \C^2)$ be the product state given by
\begin{equation}\label{eq:App:rhodiagonal}
\rho^{(diag)}=\bigotimes_{j=1}^n \begin{pmatrix}
                          \eta_j & 0 \\
                          0 & 1-\eta_j \\
                        \end{pmatrix},
\end{equation}
where $\eta_j\in [0,1]$ for $1\leq j\leq n$. Then $\rho^{(diag)}$ is quasi free with respect to the Jordan-Wigner Fermionic system $\mathcal{C}$. Also,
\begin{equation} \label{prodcormatrix}
\Gamma_{\rho^{(diag)}}^{\mathcal{C}} = \bigoplus_{j=1}^n\begin{pmatrix}
                                                          1-\eta_j & 0 \\
                                                          0 & \eta_j \\
                                                        \end{pmatrix}
\end{equation}
and
\begin{equation} \label{prodentropy}
\Tr \rho^{(diag)} \log \rho^{(diag)} = \tr \Gamma_{\rho^{(diag)}}^{\mathcal{C}} \log \Gamma_{\rho^{(diag)}}^{\mathcal{C}}.
\end{equation}
\end{proposition}

The strategy of the following proof is essentially the `classical' argument for the proof of Wick's rule for thermal states in free Fermion systems, e.g.\ \cite{BratRob}.

\begin{proof}
We first note that the formulas (\ref{prodcormatrix}) and (\ref{prodentropy}) follow from explicit calculations.
For the rest of the proof we will drop the superscript $(diag)$.
For odd values of $m$ in formula (\ref{quasifree}), the result follows from proving
\begin{equation}\label{eq:prodc=0}
\left\langle \prod_{j=1}^m c_{r_j}^{\#_{r_j}}\right\rangle_\rho=0,
\end{equation}
where the symbols $\#_{r_j}$ stand for $\#$ or nothing.
We will assume that the product of $c^{\#}$'s is not zero, otherwise there is nothing to prove. Since
\begin{equation}
a a^* a=a, \ \ a^* a a^*=a^*,\ \ \text{and}\ \ a^\#\sigma^Z= \pm a^\#
\end{equation}
it is easy to see that
\begin{equation}
\prod_{j=1}^m c_{r_j}^{\#_{r_j}}= \pm \bigotimes_{j=1}^n A_j,
\end{equation}
where $A_j\in\{a_j,a_j^*,a_ja_j^*, a_j^*a_j,\sigma^Z_j,\I\}$, and since $m$ is odd there exists $j_0\in\{1,2,\ldots,n\}$ such that $A_{j_0}\in\{a_{j_0},a_{j_0}^*\}$. Then
\begin{equation}
\Tr \prod_{j=1}^m c_{r_j}^{\#_{r_j}} \rho= \pm \prod_{j=1}^n \left(\tr A_j \begin{pmatrix}
                                                               \eta_j & 0 \\
                                                               0 & 1-\eta_j \\
                                                             \end{pmatrix}\right),
\end{equation}
which vanishes because
\begin{equation}
\tr a^{\#}\begin{pmatrix}
            \eta_{j_0} & 0 \\
            0 & 1-\eta_{j_0} \\
          \end{pmatrix} =0.
\end{equation}
The proof for even $m$ is more involved. First, assume that $\eta_j\notin \{0,1\}$ for all $1\leq j\leq n$. Note that
\begin{equation}
c_k \rho=\frac{\eta_k}{1-\eta_k}\rho c_k,
\end{equation}
because
\begin{eqnarray}
\sigma^Z\begin{pmatrix}
  \eta_j & 0 \\
  0 & 1-\eta_j \\
\end{pmatrix} &=&\begin{pmatrix}
                     \eta_j & 0 \\
                     0 & -(1-\eta_j) \\
                   \end{pmatrix}=\begin{pmatrix}
                                       \eta_j & 0 \\
                                        0 & 1-\eta_j \\
                                  \end{pmatrix}\sigma^Z, \\
a\begin{pmatrix}
       \eta_k & 0 \\
       0 & 1-\eta_k \\
     \end{pmatrix}&=&\begin{pmatrix}
                     0 & 0 \\
                     \eta_k & 0 \\
                   \end{pmatrix}=\frac{\eta_k}{1-\eta_k}\begin{pmatrix}
                                                            \eta_k & 0 \\
                                                            0 & 1-\eta_k \\
                                                          \end{pmatrix}a\,.
\end{eqnarray}
Also,
\begin{equation} \label{cfequ}
c(f)\rho=\sum_{k=1}^n \overline{{f}_k} c_k\rho=\rho \sum_{k=1}^n \overline{{f}_k} \frac{\eta_k}{1-\eta_k}c_k = \rho\,c(D_{\xi} f),
\end{equation}
where
\begin{equation}
D_\xi := \diag\left\{\xi_j: j=1,2,\ldots,n\right\}, \quad \xi_j := \frac{\eta_j}{1-\eta_j}.
\end{equation}
Similarly,
\begin{equation}
c^*(g)\rho=\rho \,c^*(D^{-1}_\xi g).
\end{equation}
The CAR imply
\begin{equation} \label{cfcommutator}
\{ c(f), c(g)\} = 0, \quad \{c(f), c^*(g)\} = \langle f, g \rangle_{\ell^2}.
\end{equation}
With $C_j := C(f_j,g_j)$, $j=1,\ldots,n$, write
\begin{equation} \label{eq:Wicks:mainproof:firstStep}
\left\langle\prod_{j=1}^m C_j\right\rangle_{\rho}= \left\langle\left(\prod_{j=2}^{m} C_j\right)c(f_1)\right\rangle_{\rho}+\left\langle\left(\prod_{j=2}^{m} C_j\right)c^*(g_1)\right\rangle_{\rho}
\end{equation}
For the first term we calculate, using (\ref{cfequ}) and cyclicity,
\begin{eqnarray} \label{itcom}
\left\langle\left(\prod_{j=2}^{m} C_j\right)c(f_1)\right\rangle_{\rho}
&=&\left\langle c(D_\xi f_1) \prod_{j=2}^{m} C_j \right\rangle_{\rho}
=\left\langle c(D_{\xi}f_1)C_m \prod_{j=2}^{m-1}C_j \right\rangle_{\rho}  \\
&=& \left\langle\left\{c(D_\xi f_1),C_m\right\} \prod_{j=2}^{m-1} C_j \right\rangle_{\rho}-\left\langle C_m c(D_\xi f_1) \prod_{j=2}^{m-1} C_j \right\rangle_{\rho} \nonumber \\
&=&\langle D_\xi f_1,g_m\rangle_{\ell^2}\left\langle\prod_{j=2}^{m-1}C_j\right\rangle_{\rho}-\left\langle C_m c(D_\xi f_1) \prod_{j=2}^{m-1} C_j \right\rangle_{\rho}, \nonumber
\end{eqnarray}
where also (\ref{cfcommutator}) was used. We proceed by applying this argument iteratively to the second term in (\ref{itcom}), commuting $c(D_{\xi} f)$ with each of the $C_j$, and eventually conclude that (\ref{itcom}) coincides with
\begin{equation}
\sum_{k=2}^m (-1)^k\langle D_\xi f_1, g_k\rangle_{\ell^2}\left\langle\prod_{\tiny{\begin{array}{c}
                                                                                         j=2 \\
                                                                                         j\neq k
                                                                                       \end{array}
}}^{m}C_j\right\rangle_{\rho}-\left\langle\left(\prod_{j=2}^m C_j
\right)c(D_{\xi}f_1)\right\rangle_{\rho}
\end{equation}
 Defining $\tilde{f}_1:=(\I+D_\xi)f_1$, the outcome of this calculation can be rewritten as
\begin{equation} \label{eq:Wicks:mainproof:firstterm}
\left\langle\left(\prod_{j=1}^m C_j\right)c(\tilde{f}_1)\right\rangle_{\rho}=\sum_{k=2}^m (-1)^k
\langle D_\xi(\I+D_\xi)^{-1} \tilde{f}_1, g_k\rangle_{\ell^2}\left\langle\prod_{\tiny{\begin{array}{c}
                                                                                          j=2 \\
                                                                                          j\neq k
                                                                                        \end{array}
}}^{m}C_j\right\rangle_{\rho}
\end{equation}
Now
\begin{equation}
\langle D_\xi(\I+D_\xi)^{-1} \tilde{f}_1, g_k\rangle_{\ell^2} = \langle D_\eta \tilde{f}_1, g_k \rangle, \ \text{where}\ \ D_{\eta}:=\diag\{\eta_j,j=1,2,\ldots,n\}
\end{equation}
and, as we can explicitly check that $\langle c_j^* c_j \rangle_{\rho} = \eta_j$, the latter coincides with
\begin{eqnarray}
\left\langle C_k\, c(\tilde{f}_1)\right\rangle_\rho & = & \langle c^*(g_k)c(\tilde{f}_1)\rangle_{\rho} \\
& = & \sum_{j=1}^n {\overline{\tilde{f}_1}_j}{g_k}_j \langle c_j^* c_j\rangle_{\rho}= \langle D_\eta \tilde{f}_1, g_k \rangle. \nonumber
\end{eqnarray}
Thus equation (\ref{eq:Wicks:mainproof:firstterm}) becomes
\begin{equation}\label{eq:Wicks:final:1Term}
\left\langle\left(\prod_{j=1}^m C_j\right)c(\tilde{f}_1)\right\rangle_{\rho}=\sum_{k=2}^m (-1)^k
\langle  C_k \,c(\tilde{f}_1) \rangle_{\rho}\left\langle\prod_{\tiny{\begin{array}{c}
                                                                                          j=2 \\
                                                                                          j\neq k
                                                                                        \end{array}
}}^{m}C_j\right\rangle_{\rho}
\end{equation}
By applying similar steps (which we omit) to the second term of (\ref{eq:Wicks:mainproof:firstStep}), and introducing $\tilde{g}_1 = (\I + D_{\xi}^{-1})g_1$ in the process, we get
\begin{equation} \label{eq:Wicks:final:2Term}
\left\langle\left(\prod_{j=1}^m C_j\right)c^*(\tilde{g}_1)\right\rangle_{\rho}=\sum_{k=2}^m (-1)^k
\langle  C_k \, c^*(\tilde{g}_1) \rangle_{\rho}\left\langle\prod_{\tiny{\begin{array}{c}
                                                                                          j=2 \\
                                                                                          j\neq k
                                                                                        \end{array}
}}^{m}C_j\right\rangle_{\rho}
\end{equation}
By substituting the results (\ref{eq:Wicks:final:1Term}) and (\ref{eq:Wicks:final:2Term}) into (\ref{eq:Wicks:mainproof:firstStep}), we obtain Wick's rule
\begin{equation}
\left\langle\prod_{j=1}^m C_j\right\rangle_{\rho}=\sum_{j=2}^m (-1)^k \langle C_k C_1\rangle_{\rho}
\left\langle\prod_{\tiny{\begin{array}{c}
                           j=2 \\
                           j\neq k \\
                         \end{array}
}}^m C_j\right\rangle_{\rho}
\end{equation}
for the case of even $m$ (after renaming $\tilde{f}_j$ and $\tilde{g}_j$ as $f_j$ and $g_j$).

Finally, in the general case where $\eta_j\in [0,1]$ for all $j=1,2,\ldots,n$, there exists a sequence
\begin{equation}
\rho_n=\bigotimes_{j=1}^n \begin{pmatrix}
                          \eta_j^{(n)} & 0 \\
                          0 & 1-\eta_j^{(n)} \\
                        \end{pmatrix}
\end{equation}
where $\eta_j^{(n)} \notin\{0,1\}\rightarrow \eta_j$ as $n\rightarrow \infty$ and thus $\rho_n\rightarrow \rho$. Now the fact that $\rho_n$ is quasi free with respect to $\mathcal{C}$ carries over to the limit.
\end{proof}

We can now prove the following fundamental relation, which was stated earlier as the second part of Lemma~\ref{lem:quasifree}.

\begin{lem}\label{lem:corr:symspectrum}
Let $\mathcal{D}$ be a Fermionic system and $\rho$ is self adjoint on $\mathcal{B}(\mathcal{H})$, then the correlation matrix $\Gamma_{\rho}^{\mathcal{C}}$ has a symmetric spectrum around $\frac{1}{2}$ and it is diagonalizable by a Bogoliubov Matrix.
\end{lem}
\begin{proof}
It is enough to show that
\begin{equation}\label{eq:corr:pf:Tdef}
T:=2\Gamma_{\rho}^{\mathcal{D}}-\I
\end{equation}
is diagonalizable by a Bogoliubov matrix. Let us define the unitary matrix $\Omega:=\frac{1}{\sqrt{2}}\begin{pmatrix}
         1 & 1 \\
         -i  & i \\
       \end{pmatrix}^{\oplus n}$.
One can check easily that $-i\Omega T \Omega^*=:\Gamma$
is real anti-symmetric matrix, i.e. $\Gamma=-\Gamma^t$.
From the spectral theory of anti-symmetric matrices, there exists an orthogonal matrix $O$ such that
\begin{equation}
\Gamma=O \bigoplus_{j=1}^n \begin{pmatrix}
                             0 &\lambda_j   \\
                              -\lambda_j & 0\\
                           \end{pmatrix} O^t,
\end{equation}
where $\lambda_j\geq 0$ for $j=1,2,\ldots,n$. By diagonalizing we get
\begin{equation}
\Omega T \Omega^*=i \Gamma
= \frac{1}{2}O  \begin{pmatrix}
                                    i & -i \\
                                    1 & 1 \\
                                  \end{pmatrix}^{\oplus n}
                                  \bigoplus_{j=1}^n \begin{pmatrix}
                                                       \lambda_j & 0 \\
                                                       0 & -\lambda_j \\
                                                     \end{pmatrix}
                                  \begin{pmatrix}
                                    -i & 1 \\
                                    i & 1 \\
                                  \end{pmatrix}^{\oplus n}O^t.
\end{equation}
Thus
\begin{equation}\label{eq:corr:pf:T}
T=W \begin{pmatrix}
                      \lambda_j & 0 \\
                      0 & -\lambda_j \\
                    \end{pmatrix}^{\oplus n} W^* \ \text{where}\ W=\frac{1}{\sqrt{2}}\Omega^* O  \begin{pmatrix}
      i & -i\\
       1 & 1 \\
     \end{pmatrix}^{\oplus n}.
\end{equation}
Note that $W$ is unitary
and one can easily check that (\ref{eq:Bogo:WJ}) is satisfied.
This proves that $W$ is a Bogoliubov matrix. Then, (\ref{eq:corr:pf:Tdef}) and (\ref{eq:corr:pf:T}) imply that
\begin{equation}
\Gamma_{\rho}^{\mathcal{D}}=W \bigoplus_{j=1}^n \begin{pmatrix}
                                                  \frac{1+\lambda_j}{2} & 0 \\
                                                  0 & \frac{1-\lambda_j}{2} \\
                                                \end{pmatrix}
W^*
\end{equation}
which shows that the spectrum is symmetric around $\frac{1}{2}$.
\end{proof}

\begin{proposition}\label{thm:App:main}
Let $\mathcal{D}$ be a Fermionic system in $\mathcal{H}$ and $\rho$ a state in $\mathcal{H}$ which is quasi free with respect to $\mathcal{D}$. Then
\begin{equation} \label{eq:App:main}
\Tr \rho\log\rho=\tr \Gamma^{\mathcal{D}}_{\rho}\log \Gamma^{\mathcal{D}}_{\rho}.
\end{equation}
\end{proposition}

\begin{proof}
Lemma \ref{lem:corr:symspectrum} implies that there exists a Bogoliubov matrix $W$ and real numbers $\eta_j$ such that
\begin{equation}
W^* \Gamma_{\rho}^{\mathcal{D}}W=\bigoplus_{j=1}^n\begin{pmatrix}
                                                    1-\eta_j & 0 \\
                                                    0 & \eta_j \\
                                                  \end{pmatrix}
=: D_{\eta}
\end{equation}
In fact we have $\eta_j \in [0,1]$: For the Fermionic system $\tilde{\mathcal{D}} = W \mathcal{D}$, we get by (\ref{cormattransfer}) that $\Gamma_{\rho}^{\tilde{\mathcal{D}}} = D_{\eta}$ and thus $\eta_j = \Tr \tilde{d}_j^* \tilde{d}_j \rho \in [0,1]$ because $\tilde{d}_j^* \tilde{d}_j$ is an orthogonal projection.

Consider the product state $\rho^{(diag)}$ of Proposition~\ref{thm:Wicks:generalmainthm} with this choice of the numbers $\eta_j$. An explicit calculation shows $\Gamma_{\rho^{(diag)}}^{\mathcal{C}} = D_{\eta}$, which is therefore equal to $\Gamma_{\rho}^{\tilde{\mathcal{D}}}$. Also, $\rho^{(diag)}$ is quasi free with respect to $\mathcal{C}$ by Proposition~\ref{thm:Wicks:generalmainthm} and $\rho$ is quasi free with respect to $\tilde{\mathcal{D}}$ by assumption and Lemma~\ref{lemma:BogoWicks1}(b). Thus, using Lemma~\ref{thm:App:corrBogo:1}, we get that $\rho$ and $\rho^{(diag)}$ are unitarily equivalent. As we have also shown that $\Gamma_{\rho}^{\mathcal{D}}$ is unitarily equivalent to $\Gamma_{\rho^{(diag)}}^{\mathcal{C}}$, (\ref{eq:App:main}) now follows from (\ref{prodentropy}).
\end{proof}



\bigskip
\begin{thebibliography}{A}

\bibitem{Aizenman94} M.~Aizenman, \emph{Localization at weak disorder: some elementary bounds}, Rev. Math. Phys. \textbf{6} (1994), 1163--1182

\bibitem{AizenmanWarzel} M.~Aizenman and S.~Warzel,
\emph{Localization bounds for multiparticle systems},
Comm. Math. Phys. \textbf{290} (2009), 903--934

\bibitem{Bachetal} V.~Bach, E.~H.~Lieb and J.~P.~Solovej, \emph{Generalized Hartree-Fock Theory and the Hubbard Model}, J. Stat. Phys. \textbf{76} (1994), 3--89

\bibitem{Bardarsonetal} J.~H.~Bardarson, F.~Pollmann and J.~E.~Moore, \emph{Unbounded growth of entanglement in models of many-body localization}, Phys. Rev. Lett. \textbf{109} (2012), 017202

\bibitem{Baskoetal} D.~M.~Basko, I.~L.~Aleiner and B.~L.~Altshuler,
\emph{Metal-insulator transition in a weakly interacting many-electron system with localized single-particle states}, Annals of Physics \textbf{321} (2006), 1126--1205

\bibitem{BauerNayak} B.~Bauer and C.~Nayak, \emph{Area laws in a many-body localized state and its implications for topological order},  J. Stat. Mech. (2013), P09005

 \bibitem{BrandHor13} F.~Brandao and M.~Horodecki,
\emph{An area law for entanglement from exponential decay of correlations},
Nature Physics \textbf{9} (2013), 721--726

\bibitem{BrandHor15} F.~Brandao and M.~Horodecki, \emph{Exponential decay of correlations implies area law}, Comm. Math. Phys. \textbf{333} (2015), 761--798


\bibitem{BratRob} O.~Bratteli and D.~Robinson, \emph{Operator algebras and quantum statistical mechanics}, Vol.~\textbf{2}, 2nd ed., New York, NY, Springer Verlag, 1997

\bibitem{BurrellOsborne}
C.~Burrell and T.~Osborne,
\emph{Bounds on the speed of information propagation in disordered quantum spin chains},
Phys. Rev. Lett. \textbf{99} (2007), 167201

\bibitem{Canovietal}
E.~Canovi, D.~Rossini, R.~Fazio, G.~E.~Santoro and A.~Silva,
\emph{Quantum Quenches, Thermalization and Many-Body Localization},
Phys. Rev. B \textbf{83} (2011), 094431


\bibitem{ChapmanStolz} J.~Chapman and G.~Stolz,
\emph{Localization for random block operators related to the XY spin chain},
Ann. Henri Poincar\'e \textbf{16} (2015), 405--435

\bibitem{ChulaevskySuhov1}
V.~Chulaevsky and Y.~Suhov,
\emph{Eigenfunctions in a two-particle Anderson tight binding model},
Comm. Math. Phys. \textbf{289} (2009), 701--723

\bibitem{ChulaevskySuhov2}
V.~Chulaevsky and Y.~Suhov,
\emph{Multi-particle Anderson localisation: induction on the number of particles},
Math. Phys. Anal. Geom. \textbf{12} (2009), 117--139

\bibitem{EisertCramerPlenio} J.~Eisert, M.~Cramer and M.~B.~Plenio, \emph{Area laws for the entanglement entropy --- a review}, Rev. Mod. Phys. \textbf{82} (2010), 277


\bibitem{ESS} A.~Elgart, M.~Shamis, and S.~Sodin,
\emph{Localisation for non-monotone Schroedinger operators},
J. Eur. Math. Soc. \textbf{16} (2014), 909Ð924

\bibitem{Friesdorfetal} M.~Friesdorf, A.~H.~Werner, W.~Brown, V.~B.~Scholz and J.~Eisert, \emph{Many-body localisation implies that eigenvectors are matrix-product states}, Phys. Rev. Lett. \textbf{114} (2015), 170505

\bibitem{GogolinEisert}
C.~Gogolin and J.~Eisert, \emph{Equilibration, thermalisation, and the emergence of statistical mechanics in closed quantum systems}, arXiv:1503.07538 (2015)

\bibitem{Gornyietal}
I.~V.~Gornyi, A.~D.~Mirlin and D.~G.~Polyakov,
\emph{Interacting electrons in disordered wires: Anderson localization and low-$T$ transport},
Phys. Rev. Lett. \textbf{95} (2005), 206603


\bibitem{HamzaSimsStolz} E.~Hamza, R.~Sims and G.~Stolz,
\emph{Dynamical localization in disordered quantum spin systems},
Comm. Math. Phys. \textbf{315} (2012) 215--239

\bibitem{Hastings} M.~Hastings, \emph{An area law for one dimensional quantum systems},
J. Stat. Mech. Theory Exp. (2007), P08024

\bibitem{HuseOganesyan} D.~A.~Huse, R.~Nandkishore and V.~Oganesyan, \emph{Phenomenology of fully many-body-localized systems}, Phys. Rev. B \textbf{90} (2014), 174202

\bibitem{Imbrie} J.~Z.~Imbrie, \emph{On Many-Body Localization for Quantum Spin Chains}, arXiv:1403.7837 (2014)

\bibitem{Itsetal}  A.~R.~Its, F.~Mezzadri and M.~Y.~Mo, \emph{Entanglement entropy in quantum spin chains with finite range interaction}, Comm. Math. Phys. \textbf{284} (2008), 117Ð185

\bibitem{JinKorepin} B.~Q.~Jin and V.~E.~Korepin, \emph{Entanglement, Toeplitz determinants and Fisher-Hartwig conjecture}, J. Stat. Phys. \textbf{116} (2004), 79--95

\bibitem{Kato} T.~Kato, \emph{Perturbation theory for linear operators}, Springer, Berlin-New York, 1976

\bibitem{KleinNguyen} A.~Klein and S.~T.~Nguyen,
\emph{The boostrap multiscale analysis for the multi-particle Anderson model},
J. Stat. Phys. \textbf{151} (2013), 938--973

\bibitem{KleinPerez} A.~Klein and J.~F.~Perez,
\emph{Localization in the ground-state of the one dimensional X-Y model with a random
transverse field},
Comm. Math. Phys. \textbf{128} (1990), 99--108

\bibitem{LRV}  J.~I.~Latorre, E.~Rico and G.~Vidal, \emph{Ground state entanglement in quantum spin chains}, Quant. Inf. Comput.  \textbf{4} (2004), 48--92

\bibitem{LSM} E.~Lieb, T.~Schultz, and D.~Mattis,
\emph{Two soluble models of an antiferromagnetic chain},
Annals of Physics \textbf{16} (1961), 407--466

\bibitem{NSS1} B.~Nachtergaele, R.~Sims and G.~Stolz,
\emph{Quantum harmonic oscillator systems with disorder},
J. Stat. Phys. \textbf{149} (2012), 969--1012

\bibitem{NSS2} B.~Nachtergaele, R.~Sims and G.~Stolz,
\emph{An area law for the bipartite entanglement of disordered oscillator systems},
J. Math. Phys. \textbf{54} (2013), 042110

\bibitem{OganesyanHuse}
V.~Oganesyan and D.~A.~Huse,
\emph{Localization of interacting fermions at high temperature},
Phys. Rev. B \textbf{75} (2007), 155111


\bibitem{PalHuse}
A.~Pal and D.~A.~Huse,
\emph{The many-body localization phase transition},
Phys. Rev. B \textbf{82} (2010), 174411


\bibitem{PasturSlavin} L.~Pastur and V.~Slavin, \emph{On the Area Law for Disordered Free Fermions}, Phys. Rev. Lett. \textbf{113} (2014), 150404

\bibitem{Serbynetal2014} M.~Serbyn, M.~Knap, S.~Gopalakrishnan, Z.~Papic, N.~Y.~Yao, C.~R.~Laumann, D.~A.~Abanin, M.~D.~Lukin and E.~A.~Demler, \emph{Interferometric probes of many-body localization}, Phys. Rev. Lett. \textbf{113} (2014), 147204

\bibitem{Serbynetal2013a} M.~Serbyn, Z.~Papic and D.~A.~Abanin, \emph{Universal slow growth of entanglement in interacting strongly disordered systems}, Phys. Rev. Lett. \textbf{110} (2013), 260601

\bibitem{Serbynetal2013b} M.~Serbyn, Z.~Papic and D.~A.~Abanin, \emph{Local conservation laws and the structure of the many-body localized states}, Phys. Rev. Lett. \textbf{111} (2013), 127201

\bibitem{Simon} B.~Simon, \emph{The statistical mechanics of lattice gases}. Princeton University Press, Princeton, NJ, 1993

\bibitem{SimonTrace} B.~Simon, \emph{Trace ideals and their applications}, Mathematical Surveys and Monographs \textbf{120}, American Mathematical Society, Providence, RI, 2005

\bibitem{SimsWarzel} R.~Sims and S.~Warzel, \emph{Decay of Determinantal and Pfaffian Correlation Functionals in One-dimensional Lattices}, arXiv:1509.00450

\bibitem{Stolz} G.~Stolz,
\emph{An introduction to the mathematics of Anderson localization},
In: Entropy and the Quantum II.
Contemp. Math. \textbf{552}, pp. 71--108. Am. Math. Soc., Providence (2011)


\bibitem{VLRK} G.~Vidal, J.~I.~Latorre, E.~Rico and A.~Kitaev, \emph{Entanglement in quantum critical phenomena}, Phys. Rev. Lett. \textbf{90} (2003), 227902

\bibitem{VoskAltman} R.~Vosk and E.~Altman,
\emph{Many-body localization in one dimension as a dynamical renormalization group fixed point},
Phys. Rev. Lett. \textbf{110} (2013), 067204

\bibitem{Znidaricetal} M.~Znidaric, T.~Prosen and P.~Prelovsek,
\emph{Many-body localization in the Heisenberg XXZ magnet in a random field},
Phys. Rev. B \textbf{77} (2008), 064426


\end{thebibliography}


\end{document}